\documentclass{IEEEtran}
\usepackage{caption, subcaption}
\usepackage{empheq}
\usepackage{amsthm}
\usepackage{amssymb, latexsym, mathrsfs}
\usepackage{amsmath}
\usepackage{dsfont}
\usepackage{algorithm}
\usepackage{color}
\usepackage{cite}
\usepackage{transparent}
\usepackage{multirow}
\usepackage{float} 
\usepackage{stfloats}
\usepackage{blindtext}
\usepackage{enumitem}
\usepackage{cuted} 
\usepackage{hyperref}

\usepackage{epsfig}
\usepackage{shellesc}
\usepackage{pgfplots,pgf}
\pgfplotsset{compat=newest}
\usepackage{tikz}
\usetikzlibrary{external}

\newlength\fheight 
\newlength\fwidth 
\usepackage{tikzscale}
\usepackage[american resistor, european voltage, european current]{circuitikz}
\usetikzlibrary{arrows,shapes,positioning}
\usetikzlibrary{decorations.markings,decorations.pathmorphing,
	decorations.pathreplacing}
\usetikzlibrary{calc,patterns,shapes.geometric}
\usetikzlibrary{arrows}
\usetikzlibrary{shapes.misc}
\usetikzlibrary{arrows.meta,positioning}

\theoremstyle{plain}
\newtheorem{theorem}{Theorem}[section]
\newtheorem{definition}{Definition}[section]

\newtheorem{lemma}{Lemma}[section]
\newtheorem{assum}{Assumption}[section]
\newtheorem{remark}{Remark}[section]

\theoremstyle{definition}

\newcommand{\Rset}{\mathbb{R}}

\newcommand{\BB}{{\mathcal{B}}}
\newcommand{\CC}{{\mathcal{C}}}
\newcommand{\DD}{{\mathcal{D}}}
\newcommand{\EE}{{\mathcal{E}}}

\newcommand{\GG}{{\mathcal{G}}}
\newcommand{\HH}{{\mathcal{H}}}
\newcommand{\II}{{\mathcal{I}}}
\newcommand{\JJ}{{\mathcal{J}}}

\newcommand{\LL}{{\mathcal{L}}}
\newcommand{\MM}{{\mathcal{M}}}
\newcommand{\NN}{{\mathcal{N}}}

\newcommand{\PP}{{\mathcal{P}}}
\newcommand{\QQ}{{\mathcal{Q}}}
\newcommand{\RR}{{\mathcal{R}}}

\newcommand{\UU}{{\mathcal{U}}}
\newcommand{\VV}{{\mathcal{V}}}

\newcommand{\XX}{{\mathcal{X}}}
\newcommand{\YY}{{\mathcal{Y}}}
\newcommand{\ZZ}{{\mathcal{Z}}}

\newcommand{\imply}{\Rightarrow}                   

\newcommand{\mbf}[1]{\mathbf{#1}}                  

\newcommand{\subss}[2]{{#1}_{[#2]}}

\begin{document}

\title{Consensus-Based Current Sharing and Voltage Balancing in DC Microgrids with Exponential Loads}

\author{Pulkit Nahata, {Mustafa S. Turan}, and   Giancarlo Ferrari-Trecate%

\thanks{ Pulkit Nahata,  Mustafa S. Turan, and Giancarlo Ferrari-Trecate are with the Dependable Control and Decision group (DECODE) of the Automatic Control Laboratory, \'Ecole Polytechnique F\'ed\'erale de Lausanne, Lausanne, Switzerland.  Email addresses:    \texttt{\{pulkit.nahata, ~mustafa.turan, ~giancarlo.ferraritrecate\}@epfl.ch}}
\thanks{This work has received support from  the Swiss National Science Foundation under the COFLEX project (grant number 200021{\_}169906).}}
\maketitle
\begin{abstract}
In this work, we present a novel consensus-based secondary control scheme for current sharing and voltage balancing in DC microgrids, composed of distributed generation units, dynamic RLC lines, and nonlinear ZIE (constant impedance, constant current, and exponential) loads. Situated atop a primary voltage control layer, our secondary controllers have a distributed structure, and utilize information exchanged over a communication network to compute necessary control actions. Besides showing that the desired objectives are always attained in steady state, we deduce sufficient conditions for the existence and uniqueness of an equilibrium point for constant power loads --- E loads with zero exponent.  Our control design hinges only on the local parameters of the generation units, facilitating plug-and-play operations. We provide a voltage stability analysis, and illustrate the performance and robustness of our designs via simulations. All results hold for arbitrary, albeit connected, microgrid and communication network topologies.
\end{abstract}
 \section{Introduction}
 \label{sec:introduction}
Thrust by the growing need to leverage the benefits of renewable energy sources, to rein in climate change and electricity costs, and to guarantee safe and reliable supply to areas lacking electric infrastructure, power generation is becoming increasingly distributed. Central to this shift in the operational exemplar are microgrids (mGs), commonly recognized as small-scale electric networks integrating multitude of distributed generation units (DGUs), storage devices, and loads. Microgrids, compatible with both AC and DC operating standards, have been demonstrated to offer manifold advantages like enhanced power quality, reduced transmission losses, and capability to operate in grid-connected and islanded modes \cite{Bhaskara}. In particular, nowadays, DC microgrids (DCmGs) are gaining ground. Their mounting popularity can be ascribed to continuous advancements in power electronics, improvements in computational power of real time controllers, availability of  inherently DC electronic loads (various appliances, LEDs, electric vehicles, computers, etc.), and presence of a natural interface with renewable energy sources (for instance PV modules) and batteries\cite{Meng}. As reviewed in \cite{KUMAR}, DCmGs are now an economically viable solution for many types of residential and industrial applications such as data centers, telecom stations, fast Electrical Vehicles (EV), net-zero energy buildings, electric ships, and hybrid energy storage systems.

In Islanded DCmGs, maintaining voltage stability is crucial, for without it voltages may either breach a critical level or drop suddenly, damaging connected loads \cite{Meng}. To this aim, a primary voltage control layer is often employed for tracking desired voltage references at the point of coupling (PC). Several primary control approaches for Buck converter--interfaced low voltage DCmGs, for example, droop-based control \cite{Shafiee2014} and plug-and-play control \cite{Martinelli2018, strehle2020scalable, Nahata, Tucci2016independent}, have been proposed in the literature. Besides voltage stability, another desirable objective is current sharing, that is, DGUs must share mG loads in accordance with their current ratings. Indeed, unregulated currents may otherwise overload generators and eventually lead to an mG failure. An additional goal of voltage balancing, requiring boundedness of weighted sum of PC voltages, is often sought to complement current sharing \cite{Tucci2018}. Being blind voltage reference emulators, primary controllers are unable to attain the aforementioned objectives all by themselves. Higher-level secondary control architectures \cite{iovine2019, LaBella} are, therefore, necessary to coordinate the voltage references provided to the primary layers. 
 
Distributed, consensus-based secondary regulators guaranteeing current sharing and voltage balancing have been the subject of many recent contributions. Centralized approaches to their synthesis are proposed in \cite{Nasirian, shafiee2014distributed_b}, but are prohibitive for large-scale mGs as they require knowledge of mG topology, lines, loads, and DGUs. Indeed, temporally varying multi-node DCmGs call for scalable design procedures \cite{Nahata, Tucci2016independent}, which enable the synthesis of decentralized controllers and plug -in/-out of DGUs on the fly without spoiling the overall stability of the network. Scalable consensus-based secondary controllers discussed in \cite{Tucci2018, Zhao} remedy the limitations of centralized design schemes while allowing for DCmGs with generic topologies; but introduce a time-scale separation by abstracting primary-controlled DGUs as ideal voltage generators or first-order systems. Moreover, they work only with static power lines. Efforts to incorporate DGU dynamics and RL lines have been made in \cite{Trip, Cucuzzella2018robust}. In \cite{Cucuzzella2018robust}, a robust distributed control algorithm considering both objectives is studied; however, a suitable initialization of the controller is needed. The resistance of the DGU filter is neglected in \cite{Trip} and hence, voltage balancing cannot be guaranteed in steady state. Unlike \cite{Tucci2018,Zhao,Trip,Cucuzzella2018robust} limited to linear loads,  \cite{DePersis} presents a power consensus algorithm intended for DCmGs feeding ZIP (constant impedance, constant current, and constant impedance) loads; although DCmG dynamics are simplified, and assumptions on the existence of a suitable steady state made.
 
All the foregoing contributions exclude E (exponential) loads --- generalized static loads which cover a wide variety of physical loads like industrial motors, fluorescent lighting, pumps, fans, etc., depending upon their exponent \cite{Kundur, Romero}. We highlight that, in DCmGs catering to E loads, steady-state current sharing and voltage balancing need to be backed by certificate guarantees. This is due to the fact that these loads, inherently nonlinear in nature, may jeopardize the stability of the DCmG, for they may introduce a destabilizing negative impedance into the network; see Section \ref{sec:stability}.

 \subsection{Paper Contributions}
 
 In this paper, we build upon previous theoretical contributions on primary voltage control \cite{Nahata}, and introduce a distributed secondary control layer for proportional current sharing and weighted voltage balancing in DCmGs consisting of DGUs, loads, and interconnecting power lines. 
 
 The main technical novelties of this paper are five-fold. First, this work does away with the modeling limitations of several existing contributions. In addition to RLC lines, we consider DGU dynamics and filter resistances. Our Buck converter--interfaced DGUs are modeled after the linear, averaged state-space model \cite{Buckmodel}. On the load modeling front, we take into account nonlinear E loads, which are popularly referred to as generalized ZIP loads, and whose power consumption depends on the exponent of the PC voltage. From what we know, this work is the very first treatise of E loads in the context of DCmGs. Second, we propose a new consensus-based secondary control scheme relying on the exchange of variables with nearest communication neighbors over a connected communication network. To achieve current sharing and voltage balancing, these secondary regulators operate at the same time scale as the primary controllers while appropriately modifying primary voltage references. In spite of their distributed structure, their control design is completely decentralized, allowing for plug-and-play operations. Third, we thoroughly investigate the steady-state behavior of the DCmG under secondary control, and show that the desired goals are always attained in steady state. Since the steady-state regime is governed by the physics of the DCmG, our specific controller has no bearing on the existence of equilibria. Moreover, for the specific case of P loads --- E loads with zero exponent, we deduce sufficient conditions on the existence and uniqueness of an equilibrium point meeting secondary goals. {Such an analysis is not trivial due to the introduced nonlinearities, and entails finding solutions to DC power-flow equations constrained to a hyperplane.} To the best of our knowledge, this has not been addressed in the literature before.   Fourth, we present a voltage stability analysis of the closed-loop DCmG, which shows that stability is independent of DCmG and communication topologies, and lays out conditions on the controller gains and power consumption of E loads. Finally, to substantiate the efficacy of our controllers, we conduct realistic simulations accommodating non-ideal DGUs with nonlinear switching behavior, and abrupt load variations.
 
 The remainder of Section \ref{sec:introduction} introduces relevant preliminaries and notation. Section \ref{sec:model} recaps the DCmG model and primary voltage control. Section \ref{sec:secondary} sets forth our secondary control scheme, and details the steady-state behavior of the closed-loop DCmG in the presence of ZIE loads. Section \ref{sec:stability} houses a stability analysis, which proves the convergence to an equilibrium point simultaneously fulfilling both current sharing and voltage balancing objectives.  Simulations validating theoretical results are provided in Section \ref{sec:simulations}.  Finally, conclusions are drawn in Section \ref{sec:conclusions}. 
 
In \cite{nahata2020existence}, a preliminary version of this work, (i) only ZIP loads were dealt with, (ii) detailed steady-state analysis was not provided, (iii) stability results and proofs, including LaSalle's analysis were skipped, and (iv) elaborate simulations with non-ideal DGUs were not conducted.

 \subsection{Preliminaries and notation}
 \textit{Sets, vectors, and functions:} We let $\mathbb{R}$ (resp. $\mathbb{R}_{>0}$) denote the set of real (resp. strictly positive real) numbers. For a finite set $\mathcal{V}$, let $|\mathcal{V}|$ denote its cardinality. Given $ x \in \mathbb{R}^{n}$, $[x] \in \mathbb{R}^{n \times n}$ is the associated diagonal matrix with $x$ on the diagonal. For vectors $x,y \in \mathbb{R^N}$, the term $x^y$ represents a vector whose $i^{th}$ element is $x_i^{y_i}$. The inequality $x\leq y$ is component-wise, that is, $x_i\leq y_i,~\forall i\in 1,...,n$.  Throughout, $\textbf{1}_n$ and $\textbf{0}_n$ are the $n$-dimensional vectors of unit and zero entries, and $\mbf{0}$ a matrix of all zeros of appropriate dimensions. 
 The average of a vector $v\in\mathbb{R}^n$ is $\langle v\rangle=\frac{1}{n}\sum_{i=1}^n v_i$. We denote with $H^1$ the subspace composed by all vectors with zero average  i.e. $H^1 = \{v\in\mathbb{R}^n:\langle v\rangle = 0\}$. The space orthogonal to $H^1$ is $H_{\perp}^1$. It holds $H_{\perp}^1 =\{\alpha\mathbf{1}_n,\text{ }\alpha\in\mathbb{R}\}$ with dim$(H_{\perp}^1)=1$. 
 
 Consider the matrix $A\in\mathbb{R}^{m\times n}$, and let $A^\dagger\in \mathbb{R}^{n \times m}$ denote its pseudo inverse. With $A(\XX|\YY)$ we indicate the linear map $A:\XX\rightarrow \YY$ where $\XX$ and $\YY$ are subspaces of $\mathbb{R}^n$ and $\mathbb{R}^m$, respectively. The range and null spaces of matrix $A$ are denoted by $\mathcal{R}(A)$ and $\mathcal{N}(A)$, respectively. For a symmetric matrix $A$, the notation $A \succ 0$ (resp. $A \succeq 0$ ) represents a positive definite (resp. positive semidefinite) matrix. 
 
 \textit{Algebraic graph theory}: We denote by $\mathcal{G}(\mathcal{V},\mathcal{E},{W})$ an undirected graph, where $\mathcal{V}$ is the node set and $\mathcal{E}=(\mathcal{V}\times\mathcal{V})$ is the edge set. If a number $l \in \{1,...,|\mathcal{E}|\}$ and an arbitrary direction are assigned to each edge, the incidence matrix $B \in \mathbb{R}^{|\mathcal{V}|\times|\mathcal{E}|}$ has non-zero components: $B_{il} = 1$ if node $i$ is the source node of edge \textit{l}, and $B_{il} =-1$ if node $j$ is the sink node of edge $l$. The \textit{Kirchoff's Current Law} (KCL) can be represented as $x = B\xi$, where $x \in \mathbb{R}^{|\mathcal{V}|}$ and $\xi \in \mathbb{R}^{|\mathcal{E}|}$ respectively represent the nodal injections and edge flows. Assume that the edge $l \in \{1,...,|\mathcal{E}|\}$ is oriented from $i$ to $j$, then for any vector $V \in \mathbb{R}^{|\mathcal{V}|}$, $(B^TV)_l=V_i-V_j$. The Laplacian matrix $\mathcal{L}$ of graph $\GG$ is $\LL=BWB^T$. If the graph is connected, then $\NN(B^T)=H_{\perp}^1$ \cite{Bullo}.
 
\section{DCmG model and primary voltage control}
\label{sec:model}
In this section, we  start by reviewing our DCmG model \cite{Nahata, Tucci2016independent} comprising multiple DGUs interconnected with each other via power lines, and recall the concepts of primary voltage control.  

\textit{DCmG Model:} {The DCmG is modeled as an undirected connected graph $\mathcal{G}_{e}=(\mathcal{D},\mathcal{E})$, where $\DD=\{1,\dots,N\}$ is the node set and $\mathcal{E}\subseteq\mathcal{D}\times\mathcal{D}$ the edge set. To each node also referred to as PC is connected a DGU and a load. The interconnecting power lines are represented by the edges of $\mathcal{G}_e$. On assigning a number to each line, one can equivalently express $\mathcal{E} = \{1,\dots,M\}$ with $M$ denoting the total number of lines. Note that edge directions are arbitrarily assigned, and provide a reference system for positive currents.  We refer the reader to Figure \ref{fig:powernework} for a representative diagram of the DCmG.}      

\begin{figure}[h!]
	\centering
	\ctikzset{bipoles/length=1.2cm}
	\tikzstyle{every node}=[font=\tiny]
	\vspace{-0.1cm}
	\begin{tikzpicture}[scale=0.5]
		\tikzstyle{DGU} = [circle, draw, double, align=center, fill=red!20, draw=red]
		
		\draw (-1,1) node(D1) [DGU]  {\textbf{DGU 1} \\ \textbf{Load 1}};
		\draw (8,3) node(D2) [DGU]  {\textbf{DGU 2}\\  \textbf{Load 2}};
		\draw (3,-3) node(D3) [DGU]  {\textbf{DGU 3} \\ \textbf{Load 3}};
		\draw (3,5)  node(D4) [DGU]  {\textbf{DGU 4} \\ \textbf{Load 4}};
		\draw (8,-3)node(D5) [DGU]  {\textbf{DGU 5} \\ \textbf{Load 5}};
		
		\tikzstyle{every node}=[font=\tiny]		
		\path [-latex,thick] (D1.north east) edge [ bend left=10] node[below, sloped] {$I_1$} (D4.south west);	
		\path [-latex,thick] (D4) edge [ bend left=10] node[below, sloped] {$I_2$} (D2.north west);	
		\path [-latex,thick] (D2.south) edge [ bend left=10] node[below, sloped] {$I_3$} (D5.north);	
		\path [-latex,thick] (D5.west) edge [ bend left=10] node[below, sloped] {$I_4$} (D3.east);	
		\path [-latex,thick] (D3.north) edge [ bend left=10] node[below, sloped] {$I_5$} (D4.south);	
		\path [-latex,thick] (D1.south east) edge [ bend right=10] node[below, sloped] {$I_6$} (D3.north west);	
		\path[dashed, very thick, blue] (D1) edge [ bend left=45] (D4);	
		\path[dashed, very thick, blue] (D1) edge [ bend right=45] (D3);	
		\path[dashed, very thick, blue] (D4) edge  (D5);	
		\path[dashed, very thick, blue] (D5.east) [ bend right=30] edge  (D2.east);	
		
	\end{tikzpicture}
	\caption{A representative diagram of the DCmG  with the communication network appearing in dashed blue.}
	\label{fig:powernework}
\end{figure}
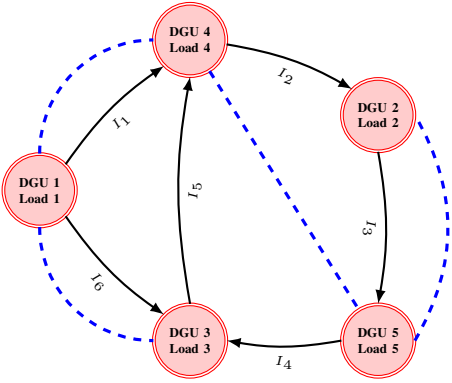
\textit{Dynamic model of a power line:} Modeled after the $\pi$-equivalent model of transmission lines \cite{Kundur}, the dynamic behavior of $l^{th}$ power line is given by
\begin{equation}
\begin{small}
\label{eq:powerline}
\subss{{\Sigma}}{l}^{Line}:\left\{\dfrac{dI_l}{dt} = - \dfrac{R_l}{L_l}I_l + \dfrac{1}{L_l}\sum_{i\in\NN_l}B_{il}V_i\right., \\
\end{small}
\end{equation} 
where $\NN_l$ is the set of DGUs incident to the $l^{th}$ line , and the variables $V_i$ and $I_l$ represent the voltage at $PC_i$ and the line current, respectively. Note that  the line capacitances are assumed to be lumped with the DGU filter capacitance $C_{ti}$. Therefore, as shown in Figure \ref{fig:ctrl_complete}, the RLC power line $l$ is equivalently represented as a $RL$ circuit with resistance $R_l>0$ and inductance $L_l>0$. 
  %
 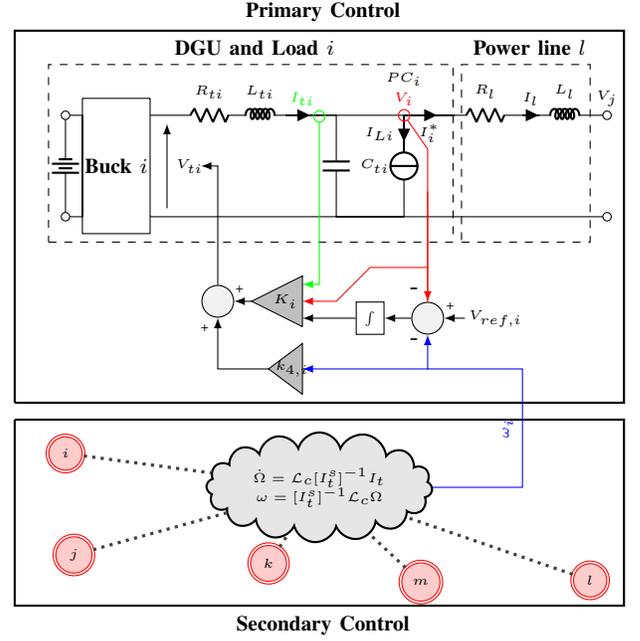
\begin{figure}[!h]
 	\centering
 	\ctikzset{bipoles/length=0.6cm}
 	\tikzstyle{every node}=[font=\tiny]
 	\vspace{0.1cm}
 	\begin{tikzpicture}[scale=0.45]
 		\draw (1,1)  to [battery, o-o](1,4)
 		to [short](1.5,4)
 		to [short](1.5,4.5)
 		to [short](3.5,4.5)
 		to [short](3.5,0.5)
 		to [short](1.5,0.5)
 		to [short](1.5,4)
 		to [short](1.5,1)
 		to [short](1,1);
 		\node at (2.5,2.5){ \footnotesize \textbf{Buck $i$}};
 		\draw[-latex] (4,1.25) -- (4,3.75)node[midway,right]{$V_{ti}$};
 		\draw (3.5,4) to [short](4,4)
 		to [short](4.5,4)
 		to [R=$R_{ti}$] (6,4)
 		to [L=$L_{ti}$] (7.5,4)
 		to [short, i=$\textcolor{green}{I _{ti}}$, -] (8.5,4)
 		to [short](9,4) 
 		to [C, l=$C_{ti}$, -] (9,1)
 		to [short](4,1)
 		to [short](3.5,1);
 		\draw (12.3,4)  to [R=$R_{l}$] (14.5,4) to [short, i=${I _{l}}$](15,4)
 		to [L=$L_{l}$] (16.5,4)
 		to [short, -o] (17,4) node[anchor=north,above]{$V_j$};
 		\draw (8.5,4) to (11,4) 
 		to [ I ] (11 ,1)
 		to [short] (9,1)
 		to [short, -o] (17,1); 
 		\draw (11,4) to [short](11.5,4);
 		\draw (11,4) node[anchor=north, above]{$\textcolor{red}{V_i}$}  to [short, i_=$I_{Li}$](11,2.9);
 		\draw (11,4) to [short, i_=$I^*_{i}$](12.5,4);
 		\node at (11,4.6)[anchor=north, above]{$PC_i$} ;
 		\draw (11,4) to (12.3,4); 
 		\draw[black, dashed] (.5,.25) -- (12.45,.25) -- (12.45,5.5) -- (.5,5.5)node[sloped, midway, above]{\footnotesize{ \textbf{DGU and Load $i$ }}}  -- (.5,.25);
 		\draw[black, dashed] (12.7,.25) -- (16.5,.25) -- (16.5,5.5) -- (12.7,5.5)node[sloped, midway, above]{\footnotesize{ \textbf{Power line $l$}}}  -- (12.7,.25);
 		\draw[red,o-] (10.9,4.15) -- (11.7,3) to (11.7,-1.5);
 		\draw[red,latex-](8,-1.5)-- (9,-1.5) --  (10,-0.5)-- (11.7,-0.5);
 		\draw[green,o-latex] (8.5,4.15) to (8.5,1.75) -- (8.5,-1) to (8,-1);
 		\draw (10,-2) node(a) [black, draw,fill=white!20] {$\normalsize{\int}$};
 		\draw[-latex] (a.west) to (8,-2);
 		\draw (11.7,-2) node(b)[ circle, draw=black, minimum size=12pt, fill=lightgray!20]{};
 		\draw[red, -latex] (11.7,1.75)  -- (b.north) node[pos=0.9, left]{\textcolor{black}{\normalsize{-}}};
 		\draw[latex-] (a.east) -- (b.west);
 		\draw[-latex] (12.8,-2)  -- (b.east) node[pos=0.7,above]{{+}} node[pos=0.25,right]{$V_{ref,i}$};
 		\draw[fill=lightgray] (8,-2.25) -- (8,-0.75) -- (6.5,-1.5) -- (8,-2.25);
 		\node at (7.5,-1.5) {$K_i$};
 		\draw[fill=lightgray] (8,-4.25) -- (8,-2.75) -- (7,-3.5) -- (8,-4.25);
 		\node at (7.7,-3.5) {$k_{4,i}$};
 		
 		\draw (5.5,-1.5) node(c)[ circle, draw=black, minimum size=12pt, fill=lightgray!20]{};
 		\draw[-latex] (6.5,-1.5) -- (c) node[pos=0.7,above]{{+}};
 		\draw[-latex] (c.north)-- (5.5,1.5) -- (5.5,2.5) -- (5,2.5);
 		\draw[-latex] (7,-3.5) -- (5.5,-3.5) --(c)node[pos=0.8,left]{\textcolor{black}{+}};
 		\draw[black, thick] (-0.5,-4.5) -- (17.5,-4.5) -- (17.5,6.5) -- (-0.5,6.5)node[sloped, midway, above]{\footnotesize{ \textbf{Primary Control}}}  -- (-0.5,-4.5);			
 		
 		\draw[black, thick] (-0.5,-5) -- (17.5,-5) -- (17.5,-10.5) --node[sloped, midway, below]{\footnotesize{ \textbf{Secondary Control}}} (-0.5,-10.5)  -- (-0.5,-5);

 		\draw[blue, -latex] (11.5,-7)  -- (14.5,-7)   -- node[sloped, midway,above]{$\omega_i$} (14.5,-3.5) -- (11.7,-3.5) -- (b.south) node[pos=0.8,left]{\textcolor{black}{\normalsize{-}}};	
 		\draw[blue, -latex] (11.7,-3.5)-- (8,-3.5);	
 		
 		\node[cloud, cloud puffs=16.2, cloud ignores aspect, minimum width=3cm, minimum height=1cm, opacity=0.75, align=center, fill=gray!20, thick, draw=black,text=black] (Gcloud) at (8.5, -7) {$\dot {\Omega} =\mathcal{L}_c[I^s_{t}]^{-1}I_t$ \\ $\omega=[I^s_{t}]^{-1}\mathcal{L}_c\Omega$};
 		
 		\tikzstyle{DGU} = [circle, draw, double, fill=red!20, draw=red]
 		
 		\draw (1,-6) node(D1) [DGU]  {\tiny $i$};
 		\draw (1.25,-9) node(D3) [DGU]  {\tiny $j$};
 		\draw (7,-9.25) node(D2) [DGU]  {\tiny $k$};
 		\draw (16.5,-9.75) node(D4) [DGU]  {\tiny $l$};
 		\draw (11.5,-9.8) node(D5) [DGU]  {\tiny $m$};
 		
 		\draw[darkgray,very thick, dotted] (D1) to  (Gcloud);
 		\draw[darkgray,very thick, dotted] (D2) to  (Gcloud);
 		\draw[darkgray,very thick, dotted] (D3) to  (Gcloud);
 		\draw[darkgray,very thick, dotted] (D4) to  (Gcloud);
 		\draw[darkgray,very thick, dotted] (D5) to  (Gcloud);
 	\end{tikzpicture}
 	\caption{Schematic diagram showing primary and secondary control layers of the DCmG, as well as the electric scheme of $i^{th}$ DGU and load. Note that the topology of the communication network is not shown. }
 	\label{fig:ctrl_complete}
 	\vspace{-.3cm}
 \end{figure}

 \textit{Dynamic model of a DGU:} The DGU comprises a DC voltage source (usually generated by a renewable resource),  a Buck converter, and a series $RLC$ filter.  The $i^{th}$  DGU, feeding a local load at $PC_i$, is connected to other DGUs via power lines.  A schematic electric diagram of the $i^{th}$ DGU along with load, connecting line(s), loads, and local PnP voltage controller is represented in Figure \ref{fig:ctrl_complete}. Upon applying KCL and KVL on the DGU side left to $PC_i$, we obtain
 \begin{equation}
 \label{eq:DGUdynamics}
 \subss{{\Sigma}}{i}^{DGU}:
 \left\{
 \begin{aligned}
 C_{ti}\dfrac{dV_{i}}{dt} &= I_{ti}-{I_{Li}(V_i,r_i)} -\sum_{l \in \EE}B_{il}I_{l}\\
 L_{ti}\dfrac{dI_{ti}}{dt} &= -V_{i}-{R_{ti}}I_{ti}+V_{ti}\\
 \end{aligned}
 \right. ,~ i\in \DD,
 \end{equation}
 where $V_{ti}$ is the command to the DC--DC Buck converter, $I_{ti}$ the filter (generator) current, and $I_{Li}(V_i,r_i)$ the current drawn by the load. The terms $R_{ti} \in \mathbb{R}_{>0}$, $L_{ti} \in \mathbb{R}_{>0}$, and $C_{ti} \in \mathbb{R}_{>0}$ are the internal resistance, capacitance (lumped with the line capacitances), and inductance of the DGU converter, respectively. 
\begin{remark}\textbf{(Modeling DC--DC converters).}
In order to work out DGU's dynamic model, we bank on the standard space averaging method \cite{Buckmodel}, enabling us to disregard the switching behavior of $V_{ti}$. Consequently, we have $V_{ti}=d_iV_{si}$, where $d_i \in [0, 1]$ is the duty cycle of the Buck converter, and $V_{si} \in \mathbb{R}$ the voltage of its power source. In this article, we suppose $V_{si}$ is large enough to avoid saturation of $d_i$.
\end{remark} 
Each DGU is equipped with a local voltage regulator, which along with other such regulators constitutes the \textit{primary control layer}. The main objective of these controllers is to ensure that the voltage at each DGU's PC tracks a reference voltage $V_{ref,i}$.  For this purpose, we augment each DGU with a multivariable PI regulator
 \begin{subequations}
 	\begin{align}
 	\dot{v}_i = \subss{e}{i} &= {V_{ref,i}}-{V}_i-\omega_i\label{eq:intdynamics},\\
 	\subss{\CC}{i}:~V_{ti} &=K_{[i]}\subss{\hat{x}}{i}+k_{4,i}\omega_i,\label{eq:controldec}
 	\end{align}
 \end{subequations}
 where $\subss{\hat{x}}{i}=\left[V_i\text{ }I_{ti}\text{ }v_i\right]^T\in\Rset^{3}$ is the state of augmented DGU,  $K_{[i]}=\left[k_{1,i}\text{ }k_{2,i}\text{ }k_{3,i}\right]\in\Rset^{1\times3}$ and $k_{4,i} \in \mathbb{R}$ are feedback gains, and $\omega_i$ is an exogenous variable generated by the \textit{secondary controller} (see Section \ref{sec:secondary} for more details). From \eqref{eq:DGUdynamics}-\eqref{eq:controldec}, the closed-loop DGU model is obtained as 
 \begin{equation}
 \label{eq:DGUdynamicsupdated}
  \begin{split}
 \subss{\hat{\Sigma}}{i}^{DGU}:
 &\left\{
 \begin{aligned}
 \dfrac{dV_{i}}{dt} &= \dfrac{1}{C_{ti}}I_{ti}-\dfrac{1}{C_{ti}}{I_{Li}(V_i,r_i)}- \dfrac{1}{C_{ti}}{I^*_i}\\
 \dfrac{dI_{ti}}{dt} &= {\alpha_i}V_{i} + {\beta_i}I_{ti}+ {\gamma_i}v_{   i}+\delta_i\omega_i\\
 \dfrac{dv_i}{dt} &= -{V}_i+{V_{ref,i}}-\omega_i
 \end{aligned}\right.\\
 \end{split},
 \end{equation}
 where
 \begin{equation}
 \label{eq:abg}
 \alpha_i=\frac{(k_{1,i}-1)}{L_{ti}},~\beta_i=\frac{(k_{2,i}-R_{ti})}{L_{ti}},~\gamma_i=\frac{k_{3,i}}{L_{ti}},
 \end{equation}
 and 
  \begin{equation}
 \label{eq:delta}
 \delta_i=\dfrac{k_{4,i}}{L_{ti}}.
 \end{equation} 
We highlight that variable $\omega_i=0$ when the secondary control layer is inactive or absent. The primary control architecture is hence decentralized as the computation of $V_{ti}$ requires only the state of $\subss{\hat{\Sigma}}{i}^{DGU}$. 

\textit{Load model:} The $i^{th}$ load is the parallel combination of Z, I, and E loads. The total current $I_{Li}(V_i,r_i)$, a function of voltage at $PC_i$, is given as 
\begin{equation}
\label{eq:loaddynamics}
I_{Li}(V_i,r_i)=\underbrace{Y_{Li}V_i}_{Z}+\underbrace{\bar{I}_{Li}}_{I} +\underbrace{V_i^{r_i-1}{P}^*_{Li}}_{E},
\end{equation}%
where $Y_{Li}$ is the conductance of the Z load while $r_i\in\mathbb{R}$ the exponent of the E load. $\bar{I}_{Li}$ and ${P}^*_{Li}$ are constants.  Note that an E load corresponds to a constant-power load when $r_i=0$, and covers wide range of physical loads depending upon the value of $r_i$. Some common examples are air conditioner {$r_i \in (0.50, 2.50)$}, resistance space heater ($r_i=2$), and fluorescent lighting ($r_i\in (1,3)$)  \cite{Kundur, Romero}.
%
%
%
\begin{assum}
	\label{ass:voltage}
	The reference signals $V_{ref,i}$ and PC voltages $V_i$ are strictly positive for all $t \geq 0$.
\end{assum}
We remark that Assumption \ref{ass:voltage} is not a limitation, and rather reflects a common constraint in microgrid operation.  Notice that, in Figure \ref{fig:ctrl_complete}, one end of the load is connected to the PC and the other to the ground, assumed be at zero potential by convention. Since the electric current and hence power flows from higher to lower potential, negative references and PC voltages will reverse the role of loads and make them power generators. In order to ensure power balance in the network, the generators will have to absorb this surplus power. This, in effect, defeats the fundamental goal of the mG, that is, the satisfiability of the loads by virtue of the power generated by the DGUs. Furthermore, if $V_i,V_{ref,i} \in \mathbb{R}^N$, then a zero-crossing for the voltages may take place. As voltages tend to zero, the power consumed by the ZIE loads with exponents $r_i<1$ approaches infinity.
 \section{Secondary control in DCmGs}
 \label{sec:secondary}
 \subsection{Problem formulation}
 The primary control layer is designed to track a suitable reference voltage $V_{ref,i}$ at the $PC_i$. As such, they do not ensure proportional current sharing and voltage balancing, defined as follows. 
 \begin{definition} \textbf{(Current sharing \cite{Tucci2018, Zhao}).}
 	\label{defn:cs}
 	The load is said to be shared proportionally among DGUs if
 	\begin{equation}
 	\label{eq:cs_defn}
 	\frac{I_{ti}}{I_{ti}^s} = \frac{I_{tj}}{I_{tj}^s}\hspace{4mm}\text{for all }i,j\in \mathcal{V},
 	\end{equation}
 	where $I_{ti}^s>0$ is the rated current of $DGU_i$.
 \end{definition}
 Current sharing ensures proportional sharing of loads amongst multiple DGUs, avoiding situations of DGU overloading, and preventing harm to the converter modules. As will be shown in the subsequent sections, the steady state voltages need not necessarily be equal to $V_{ref,i}$ when currents are shared proportionally. It is, however, desirable that PC voltages remain close to the nominal reference voltages for normal operation of the DCmG. To this aim, we state the objective of weighted voltage balancing in the following definition.
 \begin{definition}\textbf{(Weighted voltage balancing \cite{Cucuzzella2018robust}).}
 	\label{defn:vb} The voltages are said to be balanced in the steady state if 
 	\begin{equation}
 	\label{eq:vb_defn}
 	\langle[I^s_{t}]{V}\rangle =\langle[I^s_{t}]V_{ref}\rangle, 
 	\end{equation}
 	with  $V_{ref} \in \mathbb{R}^N$ being the vector of reference voltages.
 \end{definition}
 \vspace{.2cm}
 Voltage balancing implies that the weighted sum of PC voltages is equal to the the weighted sum of voltage references, ensuring boundedness of DCmG voltages. As noticed in \cite{Zhao}, in its absence, the PC voltages may experience drifts and increase monotonically despite the filter currents' being shared proportionally.
 \subsection{Consensus-based secondary control}
To achieve the aforementioned objectives, we use a consensus-based secondary control layer. Consensus filters are commonly employed for achieving
 global information sharing or coordination through distributed computations \cite{olfati2004consensus, Bullo}. In our case, we propose the following 
 consensus scheme 
 \begin{equation}
 \label{eq:basic_consensus}
 \dot {\Omega}_i =\sum\limits_{j=1, j\neq i} ^{N}a_{ij}\left(\frac{I_{ti}}{I_{ti}^s}-\frac{I_{tj}}{I_{tj}^s}\right),
 \end{equation}
  \begin{equation}
 \label{eq:omega}
 \omega_i=\frac{1}{I_{ti}^s}\sum\limits_{j=1, j\neq i} ^{N}a_{ij}\left({\Omega_i}-{\Omega_j}\right),
 \end{equation}
 where $a_{ij}>0$ if DGUs $i$ and $j$ are connected by a communication link ($a_{ij}=0$, otherwise). The corresponding \textit{communication graph} (see Figure \ref{fig:powernework}), assumed to be undirected and connected, is $\mathcal{G}_c=(\mathcal{D}, \mathcal{E}_c, W_c)$ where $(i,j)\in\mathcal{E}_{c}\Longleftrightarrow a_{ij}> 0$ and $W_c =\text{diag}\{a_{ij}\}$. Note that the topology of $\mathcal{G}_c$ and $\mathcal{G}_{e}$ can be completely different. As shown in Figure \ref{fig:ctrl_complete}, the consensus variable $\omega_i$
 modifies the primary voltage controllers; see \eqref{eq:intdynamics} and \eqref{eq:controldec}.
\begin{remark}\textbf{(Structure of secondary voltage regulators).}
	The proposed controllers have a distributed control structure, and exchange $\Omega_i$ and  $I_{ti}$ with their communication neighbors. Utilizing the received information, the $i^{th}$ DGU simultaneously computes the variable $\omega_i$ used to adapt the voltage references $V_{ref,i}$ and the DGU command  $V_{ti}$, with a view to attaining \eqref{eq:cs_defn}-\eqref{eq:vb_defn}.  It is worth noting that the scheme discussed in this work is different from \cite{Tucci2018}, where only $I_{ti}$ is communicated, and uniquely $V_{ref,i}$ is altered. In addition, this work does not reduce DGUs to ideal voltage generators or first-order systems, and eliminates assumptions on the topology of the communication network.
\end{remark}
 The complete dynamics of the DCmG under primary and secondary control are given by \eqref{eq:powerline}--\eqref{eq:delta} along with \eqref{eq:basic_consensus}--\eqref{eq:omega}. These equations can be compactly rewritten as
 \begin{equation}
 \label{eq:globalstatespace}
 \dot{X}=\mathcal{A}X+\BB(V),
 \end{equation}
 where $X= \begin{bmatrix}
 {V}^T
 &{I_t}^T
 &{v}^T
 &{I}^T
 &{\Omega}^T
 \end{bmatrix}^T\in \mathbb{R}^{4N+M}, $
 \begin{align*}
 \begin{small}
 \mathcal{A}=\underbrace{\begin{bmatrix}
 	-C_t^{-1}Y_{L} &C_t^{-1} &\textbf{0} &-C_t^{-1}B &\textbf{0}\\
 	[\alpha] &[\beta] &[\gamma] &\textbf{0} &[\delta][I^s_{t}]^{-1}\mathcal{L}_c\\
 	-\textbf{I} &\textbf{0}  &\textbf{0}  &\textbf{0}  &-[I^s_{t}]^{-1}\mathcal{L}_c\\
 	L^{-1}B^T &\textbf{0} &\textbf{0}  &-L^{-1}R &\textbf{0}\\
 	\textbf{0} &\mathcal{L}_c[I^s_{t}]^{-1}&\textbf{0}&\textbf{0}&\textbf{0}
 	\end{bmatrix}}_{\mathcal{A} \in \mathbb{R}^{(4N+M) \times (4N+M)}}
 , \end{small}
 \end{align*}
 and
 \begin{align*}
 \begin{small}
 \BB(V)=
 \underbrace{\begin{bmatrix}
 	-C_t^{-1}(\bar{I}_{L}+[V^{r-\textbf{1}_N}]P^*_{L})\\
 	\textbf{0}_N\\
 	V_{ref}\\
 	\textbf{0}_M\\
 	\textbf{0}_N
 	\end{bmatrix}}_{\mathcal{B}(V) \in \mathbb{R}^{(4N+M)}}.
 \end{small}
 \end{align*}
Note that $V \in \mathbb{R}^N$, $I_t \in \mathbb{R}^N$, $v \in \mathbb{R}^N$, $I \in \mathbb{R}^M$, ${P}^*_{L} \in\mathbb{R}^N$, $\bar{I}_{L} \in\mathbb{R}^N$, $r \in\mathbb{R}^N$, $\alpha \in \mathbb{R}^N$,  $\beta\in \mathbb{R}^N$, $\gamma \in \mathbb{R}^N$, $\delta \in \mathbb{R}^N$  are vectors of PC voltages, filter currents, integrator states, line currents, load powers, load currents, E load exponents, and parameters $\alpha_i$, $\beta_i$, $\gamma_i$, $\delta_i$ respectively. The matrices $R \in \mathbb{R}^{M \times M}_{>0}$, $L \in \mathbb{R}^{M \times M}_{>0}$, $Y_L \in \mathbb{R}^{N \times N}_{>0}$,  and $C_t \in \mathbb{R}^{N \times N}_{>0}$ are diagonal matrices collecting electrical parameters $R_l$, $L_l$, $Y_{Li}$, and $C_{ti}$, respectively. The matrix ${B} \in \mathbb{R}^{N \times M}$ is the incidence matrix of the electrical network while $\mathcal{L}_c \in \mathbb{R}^{N \times N}$  is the Laplacian matrix of the communication network.
Notice that the dynamics of the DCmG controlled only by the primary layer can be recuperated by setting $\Omega=\textbf{0}_N$ as
\begin{equation}
\label{eq:globalstatespacesanscom}
\begin{split}
\dot{X}'=&=\underbrace{\begin{bmatrix}
	-C_t^{-1}Y_{L} &C_t^{-1} &\textbf{0} &-C_t^{-1}B\\
	[\alpha] &[\beta] &[\gamma] &\textbf{0}\\
	-\textbf{I} &\textbf{0}  &\textbf{0}  &\textbf{0}\\
	L^{-1}B^T &\textbf{0} &\textbf{0}  &-L^{-1}R\\
	\end{bmatrix}}_{\mathcal{A}' \in \mathbb{R}^{(3N+M) \times (3N+M)}}
\underbrace{\begin{bmatrix}
{V}'\\
{I_t}'\\
{v}'\\
{I}'\\
\end{bmatrix}}_{\dot{X}'}\\
&+
\underbrace{\begin{bmatrix}
	-C_t^{-1}(\bar{I}_{L}+[(V')^{r-\textbf{1}_N}]P^*_{L})\\
	\textbf{0}_N\\
	V_{ref}\\
	\textbf{0}_M
	\end{bmatrix}}_{\mathcal{B}'(V) \in \mathbb{R}^{(3N+M)}}
\end{split},
\end{equation}
Indeed, \eqref{eq:globalstatespacesanscom} dictates the DCmG when the secondary layer is deactivated or in case of a communication collapse. We highlight that, for the sake of clarity, the superscript $'$ is introduced to denote DCmG states without secondary control. 

 The overall model of the DCmG \eqref{eq:globalstatespace} having been deduced, the next step is to show that the network is stable, and attains the control objectives \eqref{eq:cs_defn} and \eqref{eq:vb_defn} in the steady state. To this end, we first start by characterizing the equilibria of \eqref{eq:globalstatespace}.

\subsection{Analysis of Equilibria}
\label{sec:eq}
Before analyzing the stability of the closed-loop system \eqref{eq:globalstatespace}, we study when an equilibrium exists such that  \eqref{eq:cs_defn} and \eqref{eq:vb_defn} are jointly attained. We emphasize that, in a primary-controlled DCmG given by \eqref{eq:globalstatespacesanscom}, a reference voltage $V_{ref,i}$ is directly enforced at the $i^{th}$ PC. Thus, a unique equilibrium point  
\begin{equation}
\label{eq:equilibriumstatesanscom}
\begin{split}
\bar{X}'=&\begin{bmatrix}
V_{ref}\\
B\bar{I}+Y_LV_{ref}+[V_{ref}]^{r-\textbf{1}_N}P_L^*+\bar{I}_L\\
-[\gamma]^{-1}([\alpha] V_{ref}+[\beta]{\bar{I_t}})\\
R^{-1}B^TV_{ref}
\end{bmatrix}
\end{split},
\end{equation}
always exists \cite{Nahata}. On the contrary, once the secondary layer is activated, the voltage references are tweaked by $\omega_i$ (see \eqref{eq:intdynamics}), which is governed by equations \eqref{eq:basic_consensus} and \eqref{eq:omega}. Since the presence of exponential loads essentially renders the DCmG dynamics nonlinear, it may occur that an equilibrium point fails to exist (see Section \ref{sec:simulations} for a simulation example). Hence, in this section, we pursue whether the closed-loop system \eqref{eq:globalstatespace} possesses an equilibrium point, and if so, under what conditions on loads, topology of electrical and communication networks, and controller gains. We set off by presenting the following lemma.
\begin{lemma}
	\label{lem:equilibriumchar}
	Consider the DCmG dynamics \eqref{eq:globalstatespace}. The following statements hold: 
	\begin{enumerate}
		\item In steady state, the objectives \eqref{eq:cs_defn} and \eqref{eq:vb_defn} are attained;
		\item A steady state solution $\bar{X}=[\bar{V}^T, \bar{I}_t^T, \bar{v}^T, \bar{I}^T, \bar{\Omega}^T]^T$ exists only if there exists a $\bar{V}$ concurrently satisfying the following equations
		\begin{subequations}
			\label{eq:equilibriumexistence}
			\begin{equation}
			\mathcal{L}_e\bar{V}+\mathcal{L}_{t}[I^s_{t}]^{-1}([\bar{V}^{r-\textbf{1}_N}]P_L^*+\bar{I}_L+Y_L\bar{V}) =0, \label{eq:eqpowerflow}
			\end{equation}
			\begin{equation}
			\textbf{1}^T_N[I^s_{t}]\bar{V}=\textbf{1}^T_N[I^s_{t}]V_{ref},\label{eq:voltagebalancing}
			\end{equation}	
		\end{subequations}
		where  $\mathcal{L}_t=[I^s_{t}]-(\textbf{1}_N^T[I^s_{t}]\textbf{1}_N)^{-1}[I^s_{t}]\textbf{1}_N\textbf{1}_N^T[I^s_{t}]$, and $\mathcal{L}_e=BR^{-1}B^T$ is the Laplacian of the electrical network. 
	\end{enumerate}
\end{lemma}
\begin{proof} 
	Any steady state solution of \eqref{eq:globalstatespace} satisfies
	\begin{subequations}
		\begin{align}
		-Y_{L}\bar{V} 	-\bar{I}_{L}-[\bar{V}^{r-\textbf{1}_N}]P^*_{L}+{\bar{I}_t} -B\bar{I} &=0\label{eq:eq1}\\
		[\alpha]\bar{V} +[\beta] \bar{I_t}+[\gamma]	\bar{v} +[\delta][I^s_{t}]^{-1}\mathcal{L}_c\bar{\Omega}&=0\label{eq:eq2}\\
		V_{ref}-\bar{V} -[I^s_{t}]^{-1}\mathcal{L}_c\bar{\Omega}&=0\label{eq:eq3}\\
		B^T \bar{V}-R\bar{I} &=0\label{eq:eq4}\\
		\mathcal{L}_c[I^s_{t}]^{-1}	\bar{I_t}&=0\label{eq:eq5}
		\end{align}
	\end{subequations}
	One has from \eqref{eq:eq5} that $\bar{I}_t=\epsilon [I^s_t]\textbf{1}_N$ for some $\epsilon\in\mathbb{R}$, warranting the attainment of \eqref{eq:cs_defn}. Since $\textbf{1}^T_NB=\textbf{0}_M$, \eqref{eq:eq1} implies that $\textbf{1}_N^T\bar{I}_t=\textbf{1}_N^T(Y_{L}\bar{V} +\bar{I}_{L}+[\bar{V}^{r-\textbf{1}_N}]P^*_{L})$, then $\epsilon=(\textbf{1}_N^T[I^s_{t}]\textbf{1}_N)^{-1}\textbf{1}_N^T(Y_{L}\bar{V} +\bar{I}_{L}+[\bar{V}^{r-\textbf{1}_N}]P^*_{L}).$ We can equivalently represent 
	\begin{equation}
	\label{eq:Itbar}
	\bar{I}_t=(\textbf{1}_N^T[I^s_{t}]\textbf{1}_N)^{-1}[I^s_{t}]\textbf{1}_N\textbf{1}_N^T(Y_{L}\bar{V} +\bar{I}_{L}+[\bar{V}^{r-\textbf{1}_N}]P^*_{L}).
	\end{equation}
	Using \eqref{eq:eq4},
	\begin{equation}
	\label{eq:Ibar}
	\bar{I}=R^{-1}B^T\bar{V}.
	\end{equation}
	On substituting \eqref{eq:Itbar} and \eqref{eq:Ibar} into \eqref{eq:eq1}, one obtains \eqref{eq:eqpowerflow}. Moreover, for an $\bar{\Omega}$ to exist such that \eqref{eq:eq3} holds, $[I^s_t](V_{ref}-\bar{V})\in H^1$, which yields \eqref{eq:voltagebalancing} and guarantees \eqref{eq:vb_defn} in steady state.  If there exists a $\bar{V}$ solving \eqref{eq:equilibriumexistence}, $\bar{I}_t$ and $\bar{I}$ exist due to \eqref{eq:Itbar} and \eqref{eq:Ibar}, respectively. As \eqref{eq:voltagebalancing} holds, from \eqref{eq:eq3}, an equilibrium vector  $\bar{\Omega}=\mathcal{L}_c^\dagger[I^s_t](V_{ref}-\bar{V})+\eta\textbf{1}_N, \eta \in \mathbb{R}$ exists. Finally, on substituting $\bar{V}, \bar{I}_t, \bar{I},$ and $\bar{\Omega}$ into \eqref{eq:eq2}, one has $\bar{v}=[\gamma]^{-1}\left(([\alpha]+[\delta])\bar{V}-[\delta]V_{ref} +[\beta] \bar{I}_t\right)$.
\end{proof}
Note that equations \eqref{eq:eqpowerflow}--\eqref{eq:voltagebalancing} represent the DC power-flow equations when DGU currents are shared proportionally and PC voltages balanced. These equations are governed only by the electric network Laplacian $\LL_e$, ZIE load parameters, DGU rated currents $I^s_t$, and voltage references $V_{ref}$. We conclude that the communication network Laplacian $\LL_c$ and the controller \eqref{eq:intdynamics}--\eqref{eq:controldec} has no bearing on their solvability. In the ensuing discussion, we analyze the existence of a voltage solution to \eqref{eq:equilibriumexistence} when $r=\textbf{0}_N$, that is, the exponential loads behave as P loads.  By setting $r=\textbf{0}_N$, one can rewrite equation \eqref{eq:equilibriumexistence} as 
\begin{equation}
\label{eq:PFmatrix}
\tilde{\mathcal{L}}V=\tilde{I}-\tilde{\LL}_t[V^{-1}]P_L^*,
\end{equation}
where $\tilde{\mathcal{L}}=\begin{bmatrix}
\LL_p\\
\textbf{1}^T_N[I^s_{t}]
\end{bmatrix}$, $\LL_p=\mathcal{L}_e+\mathcal{L}_{t}[I^s_{t}]^{-1}Y_L$, $\tilde{I}=\begin{bmatrix}-\mathcal{L}_{t}[I^s_{t}]^{-1}\bar{I}_L\\
\textbf{1}^T_N[I^s_{t}]V_{ref}
\end{bmatrix}$, and $\tilde{\LL}_t=\begin{bmatrix}\mathcal{L}_{t}[I^s_{t}]^{-1}\\
\textbf{0}\end{bmatrix}$. 
\begin{remark}\textbf{(Solvability of \eqref{eq:PFmatrix}).}
The existence and uniqueness of solutions of power-flow equations have been tackled in \cite{simpson2016voltage, LaBella}. As shown in what follows, the tools therein cannot be directly applied to ascertain the solvability of \eqref{eq:PFmatrix} as \eqref{eq:voltagebalancing} restricts the voltage solutions onto a hyperplane. 
\end{remark} 
We are now in a position to state the main result.
\begin{theorem}
	\label{thm:existence}
	\textbf{(Existence and uniqueness of a voltage solution).} Consider \eqref{eq:PFmatrix} along with the vector $V^*=\tilde{\LL}^\dagger\tilde{I}$. Assume that $[V^*]$ is invertible and define $P_{cri}=4[V^*]^{-1}\tilde{\LL}^\dagger\tilde{\LL}_t[V^*]^{-1}$. Assume that the network parameters and loads satisfy 
	\begin{equation}
	\label{eq:Criticalpower}
	\Delta=||P_{cri}P_L^*||_{\infty}<1,
	\end{equation}
	and define the percentage deviations $\delta_{-}\in[0,\frac{1}{2})$ and $\delta_{+}\in(\frac{1}{2},1]$ as the unique solutions of $\Delta=4\delta_{\pm}(1-\delta_{\pm})$. The following statements hold:
	\begin{enumerate}
		\item[1)] There exists a unique voltage solution $V \in \HH(\delta_{-})$ of \eqref{eq:PFmatrix}, where
		\begin{equation}
		\label{eq:existenceset}
		\HH(\delta_{-})\coloneqq \{V\in \mathbb{R}^N |  (1-\delta_{-})V^*\leq V \leq(1+\delta_{-})V^*\}.
		\end{equation}
		Moreover, there exist no solutions of \eqref{eq:PFmatrix} in the open set 
		\begin{equation}
		\label{eq:nonexistenceset}
		\II\coloneqq \{V\in \mathbb{R}^N |  ( V >(1-\delta_{+})V^* ~\text{and}~V\notin \HH(\delta_{-}) \};
		\end{equation}
		\item[2)]  For $P_L^*=0$, $V^*$ is the unique solution of \eqref{eq:PFmatrix};
		\item [3)] If $(1-\delta_{+})V^* <V_{ref}$, then, there exist no solutions of \eqref{eq:PFmatrix} in the closed set 
		\begin{equation}
		\label{eq:nonexistenceset1}
		\JJ\coloneqq \{V\in \mathbb{R}^N |  ( V \leq (1-\delta_{+})V^* \}.
		\end{equation} 
	\end{enumerate}
\end{theorem}
\begin{proof}
	Any voltage solution to \eqref{eq:PFmatrix} must verify  $\tilde{I}-\tilde{\LL}_t[V]^{-1}P_L^*\in \RR(\tilde{\LL})$. We therefore start by characterizing the column space of $\tilde{\LL}\in \mathbb{R}^{(N+1) \times N}$. Let $\{l_1, l_2, \cdots , l_N \}, l_i \in \mathbb{R}^{N+1}$ be its column vectors. Therefore, 
	$$
	\begin{aligned}
	\RR(\tilde{\LL})&=\left\{ \sum_{i=1}^{N}a_il_i ~|~a_i \in \mathbb{R} \right\}\\
	&=\left\{ \underbrace{\begin{bmatrix}\LL_p \\ \textbf{0}\end{bmatrix}a}_{c_1}+(\sum_{i=1}^{N}I^s_{ti}a_i) \underbrace{\begin{bmatrix}\textbf{0}_N \\ 1\end{bmatrix}}_{c_2}~|~a\in\mathbb{R}^N,a_i \in \mathbb{R} \right\}
	\end{aligned}.
	$$
	The vectors $c_1$ and $c_2$ are orthogonal to each other.  As Lemma \ref{lem:rangeLE} establishes that $\RR((\LL_p))=H^1$, the vector $c_1$ can be equivalently written as 
	$$c_1=\sum_{i=1}^{N-1}\tilde{a}_i\underbrace{\begin{bmatrix}h_i\\0\end{bmatrix}}_{\tilde{h}_i},~\tilde{a}_i\in\mathbb{R},$$
	where $\{{h}_1\cdots{h}_{N-1}\},~h_i\in \mathbb{R^N}$ is an orthogonal basis of $H^1$. Hence, 
	$$\RR(\tilde{\LL})=\left\{ \sum_{i=1}^{N}\tilde{a}_i\tilde{h}_i ~|~\tilde{a}_i \in \mathbb{R} \right\},$$
	where $\tilde{h}_N=c_2$. Moreover, $\{\tilde{h}_1\cdots\tilde{h}_{N}\}$ is an orthogonal basis of $\RR(\tilde{\LL})$. Using the deduced basis, one can easily verify that $\tilde{I}\in \RR(\tilde{\LL})$ and $\RR(\tilde{\LL}_t) \subset \RR(\tilde{\LL})$. It, therefore, holds that $\tilde{I}-\tilde{\LL}_t[V]^{-1}P_L^*\in \RR(\tilde{\LL})~ \forall~ V\in \mathbb{R}^N$.
	Furthermore, we note that $\tilde{\LL}$ is a matrix with full-column rank as $dim(\RR(\tilde{\LL}))=N$. By the fundamental theorem of linear algebra,  \cite{strang1993introduction}, $dim(\RR(\tilde{\LL^T}))=dim(\RR(\tilde{\LL}))=N$, and thus $\RR(\tilde{\LL}^T)=\mathbb{R}^N$. 
	Since the linear map $\tilde{\LL}( \RR(\tilde{\LL}^T)|\RR(\tilde{\LL}))$ is always invertible  \cite{strang1993introduction}, and as $V\in \RR(\tilde{\LL}^T)$, $\tilde{I}-\tilde{\LL}_t[V]^{-1}P_L^*\in \RR(\tilde{\LL})~ \forall~ V\in \mathbb{R}^N$, one can rewrite \eqref{eq:PFmatrix} as
	\begin{equation}
	\label{eq:PFmatrixdagger}
	\begin{split}
	V&=\tilde{\mathcal{L}}^\dagger\tilde{I}-\tilde{\mathcal{L}}^\dagger\tilde{\LL}_t[V^{-1}]P_L^*\\
	&=V^*-\tilde{\mathcal{L}}^\dagger\tilde{\LL}_t[V^{-1}]P_L^*
	\end{split},
	\end{equation}
	where $\tilde{\LL}^\dagger=(\tilde{\LL}^T\tilde{\LL})^{-1}\tilde{\LL}^T$. We highlight that $\tilde{\LL}^T\tilde{\LL}$ is always invertible for matrices with full-column rank  \cite{strang1993introduction}. On utilizing the change of variables $x\coloneqq [V^*]^{-1}[V]-\textbf{1}_N$, we obtain the equivalent representation of \eqref{eq:PFmatrixdagger} as 
	\begin{subequations}
		\label{eq:PFmatrixdagger1}
		\begin{align}
		x=f(x)&\coloneqq -[V^*]^{-1}\tilde{\mathcal{L}}^\dagger\tilde{\LL}_t[V^*]^{-1}[P^*_L]r(x) \label{eq:PFa}\\
		&=-\frac{1}{4}P_{cri}[P_L^*]r(x)\label{eq:PFb},
		\end{align}
	\end{subequations}
	where $r(x)=\begin{bmatrix}\frac{1}{1+x_1},\cdots,\frac{1}{1+x_N}\end{bmatrix}^T$. Having transformed \eqref{eq:PFmatrixdagger} into \eqref{eq:PFb}, we can now apply the contraction mapping arguments presented in \cite{simpson2016voltage}. Statement 1) is a direct consequence of the Supplementary Theorem 1 of \cite{simpson2016voltage}.
	
	The proof of Statement 2) follows from \eqref{eq:PFmatrixdagger} and the invertibility of $\tilde{\LL}( \RR(\tilde{\LL}^T)|\RR(\tilde{\LL}))$. With the objective of proving Statement 3), we consider a voltage solution $V\in \JJ$. 
	Given that $(1-\delta_{+})V^* <V_{ref}$, any voltage solution $V\in \JJ$ can be represented as 
	$V=V_{ref}-b,~b\in \mathbb{R}^N_{>0}.$
	It is evident that a voltage solution $V\in\JJ$ to \eqref{eq:PFmatrix} must satisfy \eqref{eq:voltagebalancing}. Therefore, 
	\begin{equation}
	\label{eq:voltagebalancing1}
	\textbf{1}^T_N[I^s_{t}](V_{ref}-b)=\textbf{1}^T_N[I^s_{t}]V_{ref} ~\imply~ \textbf{1}^T_N[I^s_{t}]b=0.
	\end{equation}
	Since $I^s_{t},b, \in \mathbb{R}^N_{>0}$, \eqref{eq:voltagebalancing1} never holds. This concludes the proof of statement 3).
\end{proof}
\begin{remark}
	Under the sufficient conditions provided in Theorem \ref{thm:existence}, the existence of an equilibrium point depends upon the critical power matrix $P_{cri}$ and the power absorption $P_L^*$. As pointed out in~\cite{simpson2016voltage}, one can interpret $P_{cri}$ as the sensitivity of PC voltages to variations in power absorption by $P$ loads. Note that $P_{cri}$ is defined by the electrical topology of the DCmG network, the Z and I components of loads, and the voltage $V^*$ appearing at PCs when $P_L^*=0$. Clearly, from \eqref{eq:PFmatrix}, the communication network topology $\GG_c$ has no impact on $P_{cri}$. We also note that \eqref{eq:Criticalpower} is easier to satisfy for small values of $P_{Li}^*$.
\end{remark}

\begin{remark}
	The set $\HH(\delta_{-})$ in Statement~1, Theorem  represents a set where a unique voltage solution $V$ to \eqref{eq:PFmatrix} lies, whereas $\II$ is a set around $\HH(\delta_{-})$ where no solution exists. We note that the definitions and implications of these two sets resemble those of the \textit{secure solution} and \textit{solutionless} sets in \cite{simpson2016voltage}. Although, in our case, the variables $V^*$ and $P_{cri}$ defining these sets are different from \cite{simpson2016voltage}.  Moreover, we point out that as $\Delta\rightarrow0$, $\delta_{-}\rightarrow 0$ and $\delta_{+}\rightarrow1$, implying that $\HH(\delta_{-})$ converges to $\{V^*\}$ and $\II$ to the positive orthant of $\mathbb{R}^N$. On the contrary, as $\Delta \rightarrow 1$, $\delta_{-}\rightarrow \frac{1}{2}$ and $\delta_{+}\rightarrow \frac{1}{2}$, meaning that the set $\HH(\delta_{-})$ expands and the set $\II$ shrinks. We finally note that the set $\JJ$ defines a low-voltage set with no solutions under the condition $(1-\delta_{+})V^* <V_{ref}$.
\end{remark}

The previous theorem pertains to the existence of an equilibrium point for the closed-loop system \eqref{eq:globalstatespace}, albeit feeding only ZIP loads. A detailed analysis of \eqref{eq:equilibriumexistence} with $r_i \in \mathbb{R}$ is deferred to future research, as it calls for a study on finding analytic solutions to polynomials of generic order. For E loads, we will rely on the following assumption. 
\begin{assum}
	\label{assum:equilibriumpoint}
	A positive voltage solution $V\in \mathbb{R}^N_{>0}$ simultaneously satisfying equations \eqref{eq:eqpowerflow}-\eqref{eq:voltagebalancing} exists.
\end{assum}
\section{Stability of the DC microgrid network}
\label{sec:stability}
In this section, we aim to study the stability of the closed-loop system \eqref{eq:globalstatespace}, necessary in order for the DCmG to exhibit desired steady-state behavior investigated in Section \ref{sec:eq}. We start by introducing the following Lemma.
\begin{lemma}
	\label{lemma:blockmatrices}
	Consider a symmetric block matrix
	$$\mathcal{Z}=\begin{bmatrix}
	A &B\\
	B &D\\
	\end{bmatrix} \in \mathbb{R}^{2n\times 2n},$$ where $A,B$ and $D \in \mathbb{R}^{n\times n}$ are diagonal matrices. Assume that $D$ is invertible and define the matrices 
	$$Z_i=\begin{bmatrix}
	A_{i} &B_{i}\\
	B_{i} &D_{i}\\
	\end{bmatrix} \in \mathbb{R}^{2\times 2},$$ where $A_i,B_i,$ and $D_i$ represent the $i^{th}$ diagonal element of matrices $A,B$ and $C$, respectively. Matrix $\ZZ$ is positive definite if and only if $Z_i\succ 0$ for all $i \in \{1,\dots,n\}$. If at least one $Z_i$ is positive semidefinite, then $Z$ is positive semidefinite.
\end{lemma}
\begin{proof}
	The matrix Z is positive definite if and only if $D \succ 0$, and its Schur's complement $A-BD^{-1}B \succ 0$. Considering that $A,B,C,$ and $D$ are diagonal matrices, the aforementioned conditions translate into $D_{i} > 0$, and $A_{i}-B_{i}D^{-1}_{i}B_i \succ 0,~\forall i \in \{1,\dots,n\}$. Note that $A_{i}-B_{i}D^{-1}_{i}B_{i}$ is the Schur's complement of $Z_i$. Therefore, if $Z_i \succ 0~ \forall~i \in  \{1,\dots,n\}$, $\mathcal{Z} \succ 0$. If the $i^{th}$ $Z_i$ is positive semidefinite, then 
	$$\text{det}(Z_i)=A_{i}D_{i}-B_{i}B_{i}=0.$$
	Since $D_{ii}\neq 0$, $A_{i}-B_{i}D_{i}^{-1}B_{i}=0$. This implies that diagonal entry in the $i^{th}$ row of $A-BD^{-1}B$ is equal to zero, making the Schur's complement of $\ZZ$ positive semidefinite, and hence, $\ZZ$ positive semidefinte.
\end{proof}
\begin{theorem}
	\label{thm:stability}
	\textbf{(Stability of the closed-loop DCmG).}  Consider the closed-loop system \eqref{eq:globalstatespace}, along with Assumption \ref{assum:equilibriumpoint}. Define the equilibrium power absorption of the $i^{th}$ exponential load as $\bar{P}^*_{Li}=P^*_{Li}\bar{V}_i^{r_i}$. For $i \in \DD$, if the feedback gains $k_{1,i},~k_{2,i}$, and $k_{3,i}$ belong to the set 
	\begin{equation}
	\label{eq:explicitgains}
	\ZZ_{[i]}= \left\{ \begin{array}{l}
	k_{1,i}<1,\\
	k_{2,i}<R_{ti},\\
	0<k_{3,i}<\frac{1}{L_{ti}}(k_{1,i}-1)(k_{2,i}-R_{ti})
	\end{array} \right\},
	\end{equation} $k_{4,i}=k_{1,i}-1$, and the Z and E components of the ZIE load \eqref{eq:loaddynamics} with $r_i<1$ verify
	\begin{equation}
	\label{eq:maximumpowerconsumption}
	(1-r_i)\bar{P}^*_{Li}< Y_{Li}\bar{V}_{i}^2,
	\end{equation} then the following statements hold:
	\begin{enumerate}
		\item[1)] The equilibrium point $\bar{X}$ is locally asymptotically stable, and is globally asymptotically stable when $\bar{P}^*_L=0$; 
		\item[2)] In the absence of a communication network, the equilibrium point $\bar{X}'$ of the resulting closed-loop system \eqref{eq:globalstatespacesanscom} is locally asymptotically stable.  
	\end{enumerate}
\end{theorem}
\begin{proof}
	\textit{Statement 1):} To study the behavior of trajectories resulting from \eqref{eq:globalstatespace}, consider the following candidate Lyapunov function, attaining a minimum at $\bar{X}$
	\begin{equation}
	\label{eq:globalLyapunov}
	\VV({\tilde{X}})=\frac{1}{2}{\tilde{X}}^T\mathcal{P}{\tilde{X}}, 
	\end{equation}
	where $\tilde{X}=X-\bar{X}$. The matrix $\PP$ is defined as
	\begin{equation}
	\begin{aligned}
	\label{eq:globallyapunovfunction}
	\mathcal{P}&=\left[ \begin{array}{c|c}
	\PP_1  &\textbf{0} \\
	\hline
	\textbf{0}& \PP_2
	\end{array}\right]\\&=\left[ \begin{array}{cccc|c}
	C_t & \textbf{0} & \textbf{0} & \textbf{0} &\textbf{0} \\
	\textbf{0} &[\beta][\omega]^{-1} &[\gamma][\omega]^{-1}&\textbf{0} &\textbf{0}\\
	\textbf{0} &[\gamma][\omega]^{-1} &[\alpha][\gamma][\omega]^{-1} &\textbf{0} &\textbf{0}\\
	\textbf{0} & \textbf{0} & \textbf{0} &L &\textbf{0} \\
	\hline
	\textbf{0} & \textbf{0} &\textbf{0} &\textbf{0} &\textbf{I}\\
	\end{array}\right],
	\end{aligned}
	\end{equation}
	with $\PP_1\in\mathbb{R}^{(3N+M) \times (3N+M)}$, $\PP_2\in\mathbb{R}^{N \times N}$, and $[\omega]=[\gamma]-[\alpha][\beta]\in \mathbb{R}^{N \times N}$. To ensure \eqref{eq:globalLyapunov} is a legitimate Lyapunov function, the matrix $\PP$ must be positive definite. In fact, as $\PP$ is a block diagonal matrix with $C_t\succ 0$, $L\succ 0$, and $I\succ 0$, its positive definiteness hinges on 
	$$
	\hat{\PP}=
	\begin{bmatrix}
	[\beta][\omega]^{-1} &[\gamma][\omega]^{-1}\\
	[\gamma][\omega]^{-1} &[\alpha][\gamma][\omega]^{-1}
	\end{bmatrix}\succ 0.
	$$
	which, as a direct consequence of Lemma \ref{lemma:blockmatrices}, translates into 
	$$
	\hat{\PP}_i=
	\begin{bmatrix}
	\frac{\beta_i}{\omega_i} & \frac{\gamma_i}{\omega_i}\\
	\frac{\gamma_i}{\omega_i} & \frac{\alpha_i\gamma_i}{\omega_i}
	\end{bmatrix}\succ 0, \forall i \in \DD.
	$$
	Using Sylvester's criterion \cite[Theorem 7.2.5]{Horn} followed by some basic algebra, one can deduce that $\hat{\PP}_i\succ 0$ if and only if $\beta_i,\gamma_i,\omega_i$ belong to the set $$\mathcal{S}_i=\{(\beta_i,\gamma_i,\omega_i):(\beta_i, \omega_i>0, \gamma_i<0)~\text{or}~(\beta_i, \omega_i<0, \gamma_i>0)\}.$$
The time derivative of \eqref{eq:globallyapunovfunction} along the solutions of \eqref{eq:globalstatespace} reads
	\begin{equation}
	\label{eq:lyapderivative}
	\begin{split}
	\dot{\VV}{(\tilde{X})}&=\left(\dfrac{\partial\VV}{\partial\tilde{X}}\right)^T\dot{\tilde{X}}\\
	&=\frac{1}{2}\left({X}^T\mathcal{A}^T\mathcal{P}\tilde{X}+\tilde{X}^T\mathcal{P}\mathcal{A}X\right)+\BB(V)^T\PP\tilde{X}\\
	&=\frac{1}{2}\left(\tilde{X}^T\mathcal{A}^T\mathcal{P}\tilde{X}+\tilde{X}^T\mathcal{P}\mathcal{A}\tilde{X}\right)+\left(\mathcal{A}\bar{X}+\BB(V)\right)^T\PP\tilde{X}\\
	&={\tilde{X}}^T{\QQ}(V)\tilde{X}
	\end{split},
	\end{equation}
	where ${\QQ}(V)$ is defined in \eqref{eq:Qmatty}, and 
\begin{figure*}[b]
\hrule
	\begin{equation}
	\label{eq:Qmatty}{\QQ}(V)=-\frac{1}{2}\begin{bmatrix}
	2(Y_L+Y_E(V)) &\begin{split}&-\textbf{I}+[\gamma][\omega]^{-1}\\&-[\alpha][\beta][\omega]^{-1}\end{split} &\textbf{0} &\textbf{0} &\textbf{0}\\
	\begin{split}&-\textbf{I}+[\gamma][\omega]^{-1}\\&-[\alpha][\beta][\omega]^{-1}\end{split} &-2[\beta]^2[\omega]^{-1} &-2[\beta][\gamma][\omega]^{-1}  &\textbf{0} &\begin{split}&[I^s_t]^{-1}\left(-\textbf{I}+[\gamma][\omega]^{-1}\right.\\&\left.-[\alpha][\beta][\omega]^{-1}\right)\LL_c\end{split} \\
	\textbf{0} &-2[\beta][\gamma][\omega]^{-1}  &-2[\gamma]^2[\omega]^{-1}  &\textbf{0} &\begin{split}&[I^s_t]^{-1}\left([\alpha][\gamma][\omega]^{-1}\right.\\&\left.-[\delta][\gamma][\omega]^{-1}\right)\LL_c\end{split}\\
	\textbf{0} &\textbf{0} &\textbf{0} &2R &\textbf{0}\\
	\textbf{0} &\begin{split}&\LL_c\left(-\textbf{I}+[\gamma][\omega]^{-1}\right.\\&\left.-[\alpha][\beta][\omega]^{-1}\right)[I^s_t]^{-1}\end{split} &\begin{split}&\LL_c\left([\alpha][\gamma][\omega]^{-1}\right.\\&\left.-[\delta][\gamma][\omega]^{-1}\right)[I^s_t]^{-1}\end{split}  &\textbf{0} &\textbf{0}
	\end{bmatrix},
	\end{equation}
\end{figure*}
and $Y_E(V)$ is a diagonal matrix, whose $i^{th}$ diagonal element is 
	\begin{equation}
	Y_{Ei}(V_i)=\frac{\bar{P}_{Li}^*(V_i^{r_i-1}-\bar{V}_i^{r_i-1})}{\bar{V}_i^{r_i}(V_i-\bar{V}_i)}.
	\end{equation}
	Using \eqref{eq:abg} and \eqref{eq:delta}, one can simplify $\QQ(V)$ as 
		\begin{small}
	$${\QQ}(V)=-\begin{bmatrix}
	Y_L+Y_E(V) &\textbf{0} &\textbf{0} &\textbf{0} &\textbf{0}\\
	\textbf{0} &-[\beta]^2[\omega]^{-1} &-[\beta][\gamma][\omega]^{-1}  &\textbf{0} &\textbf{0} \\
	\textbf{0} &-[\beta][\gamma][\omega]^{-1}  &-[\gamma]^2[\omega]^{-1}  &\textbf{0} &\textbf{0}\\
	\textbf{0} &\textbf{0} &\textbf{0} &R &\textbf{0}\\
	\textbf{0} &\textbf{0} &\textbf{0}  &\textbf{0} &\textbf{0}
	\end{bmatrix},$$ 
		\end{small}
	To claim that $	\dot{\VV}{(\tilde{X})}\leq 0$, and subsequently the stability of the equilibrium point $\bar{X}$, one needs
	\begin{equation}
	\label{eq:timevaryingimpedance}
	f_i(V_i)=Y_{Li}+Y_{Ei}(V_i) \geq 0,~\forall i \in \DD\end{equation}
	and, from Lemma \ref{lemma:blockmatrices}, 
	$$\hat{\QQ}_i=\begin{bmatrix}
	-\frac{\beta_i^2}{\omega_i} &-\frac{\beta_i\gamma_i}{\omega_i}\\
	-\frac{\beta_i\gamma_i}{\omega_i}  &-\frac{\gamma_i^2}{\omega_i}
	\end{bmatrix}\succeq 0, \forall i \in \DD.$$
	Evidently, $\hat{\QQ}_i\succeq 0$ if and only if $\omega_i$ belongs to 
	$$\mathcal{T}_i=\{\omega_i:\omega_i<0\}.$$
	Assume for the moment that \eqref{eq:timevaryingimpedance} holds. For $\dot{\VV}{(\tilde{X})}\leq 0$ and $\VV({\tilde{X}})>0$ to be verified simultaneously, $\alpha_i$, $\beta_i$, and $\gamma_i$ should be such that $(\beta_i,\gamma_i,\omega_i)\in\mathcal{S}_i$, and $\omega_i\in\mathcal{T}_i$. Equivalently, $(\alpha_i,\beta_i,\gamma_i)$ must belong to
	\begin{equation}
	\label{eq:explicitgains1}
	\begin{split}
	\UU_{i}
	&=\{(\alpha_i,\beta_i,\gamma_i) : \alpha_i<0, \beta_i<0,0<\gamma_i<\alpha_i\beta_i\}
	\end{split}.
	\end{equation}
	Using \eqref{eq:abg}, one can rewrite set $\UU_{i}$ in terms of $k_{1,i}$, $k_{2,i}$, and $k_{3,i}$ as \eqref{eq:explicitgains}. Now, as for \eqref{eq:timevaryingimpedance}, it is state dependent and should, at least, hold at $V_i=\bar{V}_i$. Note that $f_i(V_i)$ has a finite limit for $V_i\to\bar{V}_i$, which one can show by employing Bernoulli-Hospital theorem as
	\begin{equation}
	\label{eq:continuity}
	\begin{split}
	f_i(\bar{V}_i)&=\lim_{V_i\rightarrow\bar{V}_i}f_i({V}_i)\\
	&=Y_{Li}+\lim_{V_i\rightarrow\bar{V}_i}\frac{\bar{P}_{Li}^*(V_i^{r_i-1}-\bar{V}_i^{r_i-1})}{\bar{V}_i^{r_i}(V_i-\bar{V}_i)}\\
	&=Y_{Li}-\frac{\bar{P}^*_{Li}}{\bar{V}_i^{2}}(1-r_i)
	\end{split}.
	\end{equation}
	In view of \eqref{eq:continuity}, if \eqref{eq:maximumpowerconsumption} is verified by all the ZIE loads with $r_i<1$, then the inequality \eqref{eq:timevaryingimpedance} holds in a neighborhood of $X=\bar{X}$. Note that \eqref{eq:maximumpowerconsumption} is always satisfied when $r_i\geq 1$. We can now state that a compact level set $\MM$ of $\VV({\tilde{X}})$ can be taken sufficiently small such that it is contained in the neighborhood within which \eqref{eq:timevaryingimpedance} holds. As a result, if $X(0)-\bar{X}\in \MM$, then $X-\bar{X}\in \MM$ for all $t\geq 0$. To show local asymptotic stability, one can exploit the standard LaSalle's invariance principle and show that the largest invariant set $M\subset \MM$ contains solely the equilibrium point $\bar{X}$. A detailed computation of $M$ is skipped here and presented in Appendix \ref{appendixA}.\\
	We point out that when $P_L^*=0$, one can call into use Theorem \ref{thm:existence} and Lemma \ref{lem:equilibriumchar} to establish existence and uniqueness of $\bar{X}$. Moreover, \eqref{eq:timevaryingimpedance} holds for all $X\in \mathbb{R}^{4N\times M}$. Hence, $\bar{X}$ is globally asymptotically stable.

	\textit{Statement 2):} This proof relies heavily on the preceding analysis. Consequently, instead of providing a detailed proof, we sketch the proof of Statement 2). In the absence of a communication network, the DCmG dynamics given by \eqref{eq:globalstatespacesanscom} admit a unique equilibrium; see \eqref{eq:equilibriumstatesanscom}. To show the asymptotic stability of $\bar{X}'$, consider the following Lyapunov function 	
	\begin{equation}
	\VV(\tilde{X}')=\frac{1}{2}\tilde{X}'^{T}\mathcal{P}_1{\tilde{X}'}, 
	\end{equation}
	where $\tilde{X}'=X-\bar{X}'$. One can now trace the same steps as before to reach the conclusion.
\end{proof}
\begin{remark}\textbf{(Power consumption of E loads and stability).} Based on Theorem \ref{thm:stability}, the permissible power drawn by E loads with $r_i<1$ is restricted by \eqref{eq:maximumpowerconsumption}. To make plain sense out of \eqref{eq:maximumpowerconsumption}, one can state that, just like P loads \cite{Ashourloo, Mingfei}, E loads with $r_i<1$ exhibit a negative incremental admittance ($dI/dV < 0$; see \eqref{eq:loaddynamics}) having a destabilizing impact. To preserve stability of the network, the DCmG operator needs to counter this negative damping with the positive damping of Z loads, constraining the power consumption of E loads with $r_i<1$. 	Indeed, for E loads with $r_i<1$, stability cannot be guaranteed in the absence of Z loads. Note that no upper limit exists for E loads with $r_i>1$ as  \eqref{eq:maximumpowerconsumption} is always fulfilled. 
\end{remark}
\begin{remark}\textbf{(Compatibility with primary control and stability under a communication collapse).} Equations \eqref{eq:powerline} and \eqref{eq:DGUdynamicsupdated} represent the DCmG under primary control when the secondary layer is inactive. As shown in Theorem \ref{thm:stability}, \eqref{eq:explicitgains} and \eqref{eq:maximumpowerconsumption} also make possible the design of stabilizing primary controllers, allowing us to reach the following conclusions: (i) the proposed secondary controllers are design-wise fully compatible with the primary layer, and require only an additional control gain $k_{4,i}=k_{1,i}-1, k_{1,i} \in \ZZ_{[i]}$ be set once activated; (ii) if the DCmG undergoes a communication collapse, the primary controllers maintain voltage stability without any human intervention, forcing each PC to track $V_{ref,i}$ in steady state. 
\end{remark}
\begin{remark}\textbf{(Plug-and-play nature of control).} To synthesize the proposed  secondary regulators, one can use the feedback gains \eqref{eq:explicitgains}, dependent on the DGU filter parameters $R_{ti}$ and $L_{ti}$ but not on $C_{ti}$ --- assumed to be lumped with line capacitances.  Capable of taking into account the worst-case parameter variations around a nominal value, these explicit inequalities  \eqref{eq:explicitgains} cause the entire control design to be robust to uncertainties in filter parameters. Moreover, the secondary controllers, notwithstanding their distributed structure, can be designed in a completely decentralized fashion, enabling plug-and-play operations. For example, when a new DGU is plugged-in, its controller can be designed without the knowledge of any other parameter of the microgrid, and no other controller in the microgrid needs to be updated such that voltage stability is preserved. As a last comment, we note that if a power line with non-negligible capacitance is added or removed, no DGU controller needs to be updated as controller gains are independent of $C_{ti}$ and hence the capacitance of lines. \end{remark}
\section{Simulations}
\label{sec:simulations} 

\begin{figure}[h!]
	\centering
	\ctikzset{bipoles/length=1.2cm}
	\tikzstyle{every node}=[font=\tiny]
	\vspace{-0.1cm}
	\begin{tikzpicture}[scale=0.4]
		\tikzstyle{DGU} = [circle, draw, double, align=center, fill=red!20, draw=red]
		
		\draw (-1.5,9) node(D1) [DGU]  {\textbf{DGU 1} \\ \textbf{Load 1}};
		\draw (-7.5,3) node(D2) [DGU]  {\textbf{DGU 2}\\  \textbf{Load 2}};
		\draw (4.5,3) node(D3) [DGU]  {\textbf{DGU 3} \\ \textbf{Load 3}};
		\draw (-1.5,-3)  node(D4) [DGU]  {\textbf{DGU 4} \\ \textbf{Load 4}};
		\draw (9,-3)node(D5) [DGU]  {\textbf{DGU 5} \\ \textbf{Load 5}};
		\draw (9,9)node(D6) [DGU]  {\textbf{DGU 6} \\ \textbf{Load 6}};
		
		\tikzstyle{every node}=[font=\tiny]		
		\path [-latex,thick] (D1.south east) edge [ bend left=0] node[below, sloped] {$l_2$} (D3.north west);	
		\path [-latex,thick] (D1.south west) edge [ bend right=0] node[below, sloped] {$l_1$} (D2.north east);	
		\path [-latex,thick] (D1.east) edge [ bend left=0] node[below, sloped] {$l_3$} (D6.west);	
		\path [-latex,thick] (D2.south east) edge [ bend right=0] node[above, sloped] {$l_4$} (D4.north west);	
		\path [-latex,thick] (D3.south west) edge [ bend left=0] node[above, sloped] {$l_5$} (D4.north east);	
		\path [-latex,thick] (D4.east) edge [ bend right=0] node[below, sloped] {$l_6$} (D5.west);	
		\path [-latex,thick] (D5.north) edge [ bend right=0] node[below, sloped] {$l_7$} (D6.south);	
		\path[dashed, very thick, blue] (D1) edge [ bend left=0] (D4);	
		\path[dashed, very thick, blue] (D1) edge [ bend left=15] (D3);	
		\path[dashed, very thick, blue] (D2) edge [ bend left=0]  (D3);	
		\path[dashed, very thick, blue] (D3) [ bend left=15] edge  (D4);	
		\path[dashed, very thick, blue] (D3) [ bend right=15] edge  (D5);
		\path[dashed, very thick, blue] (D3) [ bend left=15] edge  (D6);		
		
	\end{tikzpicture}
	\caption{A representative diagram of the DCmG  with the communication network appearing in dashed blue.}
	\label{fig:mG_simulation}
\end{figure}
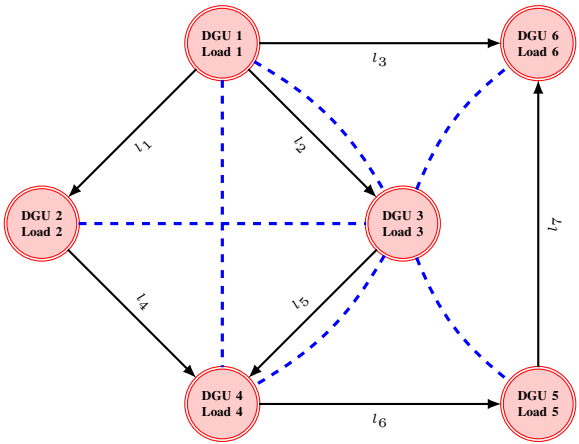

The proposed consensus-based controller is evaluated via realistic computer simulations using the \textit{Specialized Power Systems Toolbox} of Simulink. The considered DCmG has 6 DGUs arranged in the topology given in Figure~\ref{fig:mG_simulation}, where electrical lines depicted in solid black arrows are assigned arbitrary directions\footnote{Arrows define a reference frame for positive currents.} and bidirectional communication channels are shown in blue dashed lines. We further assume that power lines are equipped with switches so as to enable or interrupt power transfer. The DGUs consist of bidirectional Buck converters fed by source voltages of $V_{s,i} = 80V, \forall i \in \mathcal{N}$ as well as RLC filters and loads with non-identical parameter values. Bidirectional Buck converters are implemented as non-ideal insulated gate bipolar transistor (IGBT) switches that operate at $15$ kHz and have snubber circuits as a safeguard against large transients that can damage electrical equipments. The parameters of filters and lines are adopted from~\cite{Tucci2016independent}, whereas those of the loads are selected so as to satisfy \eqref{eq:maximumpowerconsumption}. Voltage reference values $V_{ref,i}$ are chosen to be between $45$V and $50$V, and the primary controller gains $k_{1,i}$, $k_{2,i}$, and $k_{3,i}$ are selected from the set $\mathcal{Z}_{[i]}$ in \eqref{eq:explicitgains}. 

The simulations are divided into two parts. We first present a simulation scenario verifying that the above DCmG with ZIP loads and proposed controller structure converges to the unique solution in $\mathcal{H}(\delta_{-})\cup\mathcal{I}$; see \eqref{eq:Criticalpower}, \eqref{eq:existenceset}. We then change some ZIP loads to ZIE loads to show that stability and secondary control objectives are still achieved despite the results of Theorem~\ref{thm:existence} being no longer applicable. In the second part of simulations, we show that an equilibrium fails to exist when some parameters of the DCmG equipped with ZIP loads are modified such that conditions for Theorem~\ref{thm:existence} are not satisfied. 

\subsection{Convergence to an equilibrium}
\label{subsec:SimulationConvergence}

In this scenario, we show that the proposed controller achieves current sharing and voltage balancing while allowing for plugging-in and unplugging of DGUs.

\begin{figure*}[t!]
	\definecolor{mycolor1}{rgb}{0.00000,0.44700,0.74100}%
	\definecolor{mycolor2}{rgb}{0.85000,0.32500,0.09800}%
	\definecolor{mycolor3}{rgb}{0.92900,0.69400,0.12500}%
	\definecolor{mycolor4}{rgb}{0.49400,0.18400,0.55600}%
	\definecolor{mycolor5}{rgb}{0.46600,0.67400,0.18800}%
	\definecolor{mycolor6}{rgb}{0.30100,0.74500,0.93300}%
	\centering
	\begin{tikzpicture} 
    \begin{axis}[%
    hide axis,
    xmin=10,
    xmax=50,
    ymin=0,
    ymax=0.4,
    legend style={legend columns=6, legend cell align=left, align=left, draw=white!15!black}
    ]
    \addlegendimage{mycolor1,line width = 1.5}
    \addlegendentry{\small{ $V_1$}};
    \addlegendimage{mycolor2,line width = 1.5}
    \addlegendentry{\small{ $V_2$}};
    \addlegendimage{mycolor3,line width = 1.5}
     \addlegendentry{\small{ $V_3$}};
     \addlegendimage{mycolor4,line width = 1.5}
    \addlegendentry{\small{ $V_4$}};
    \addlegendimage{mycolor5,line width = 1.5}
    \addlegendentry{\small{ $V_5$}};
    \addlegendimage{mycolor6,line width = 1.5}
     \addlegendentry{\small{ $V_6$}};
    \end{axis}
	\end{tikzpicture}
	\begin{tikzpicture} 
    \begin{axis}[%
    hide axis,
    xmin=10,
    xmax=50,
    ymin=0,
    ymax=0.4,
    legend style={legend columns=6, legend cell align=left, align=left, draw=white!15!black}
    ]
    \addlegendimage{mycolor1,line width = 1.5}
    \addlegendentry{\small{ $\frac{I_{t1}}{{I}^s_{t1}}$}};
    \addlegendimage{mycolor2,line width = 1.5}
    \addlegendentry{\small{$\frac{I_{t2}}{{I}^s_{t2}}$}};
    \addlegendimage{mycolor3,line width = 1.5}
     \addlegendentry{\small{$\frac{I_{t3}}{{I}^s_{t3}}$}};
     \addlegendimage{mycolor4,line width = 1.5}
    \addlegendentry{\small{$\frac{I_{t4}}{{I}^s_{t4}}$}};
    \addlegendimage{mycolor5,line width = 1.5}
    \addlegendentry{\small{$\frac{I_{t5}}{{I}^s_{t5}}$}};
    \addlegendimage{mycolor6,line width = 1.5}
     \addlegendentry{\small{$\frac{I_{t6}}{{I}^s_{t6}}$}};
    \end{axis}
\end{tikzpicture}\\
\definecolor{mycolor1}{rgb}{0.00000,0.44700,0.74100}%
\definecolor{mycolor2}{rgb}{0.85000,0.32500,0.09800}%
	\begin{tikzpicture} 
    \begin{axis}[%
    hide axis,
    xmin=10,
    xmax=50,
    ymin=0,
    ymax=0.4,
    legend style={legend columns=2, legend cell align=left, align=left, draw=white!15!black}
    ]
    \addlegendimage{mycolor1,dashed,line width = 1.5}
    \addlegendentry{\small{$\sum_i (V_{ref,i}-\omega_i)I^s_{ti}$}};
    \addlegendimage{mycolor2,line width = 1.5}
    \addlegendentry{\small{$\sum_i V_iI^s_{ti}$}};
    \end{axis}
\end{tikzpicture}\\
	\begin{subfigure}[t]{0.32\textwidth}
		\centering
		\includegraphics{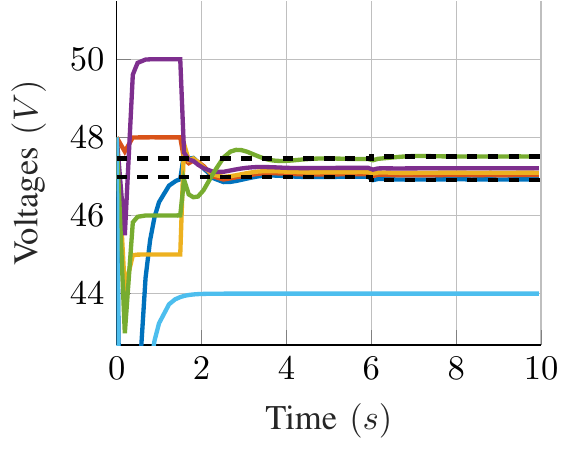}
		\caption{}
		\label{fig:ZIP_voltages}
	\end{subfigure}
	\begin{subfigure}[t]{0.32\textwidth}
		\centering
		\includegraphics{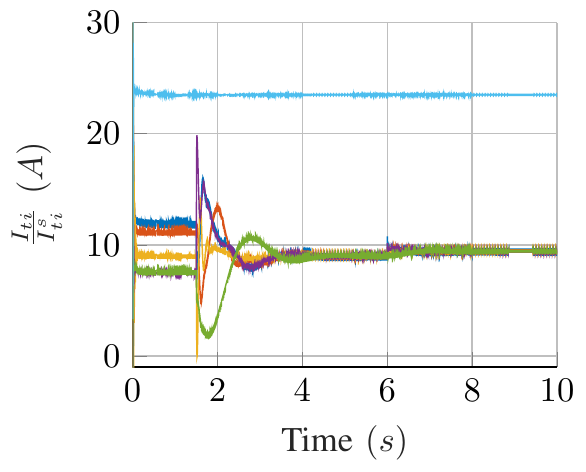}
		\caption{}
		\label{fig:ZIP_currents}
	\end{subfigure}
	\begin{subfigure}[t]{0.32\textwidth}
		\centering
		\includegraphics{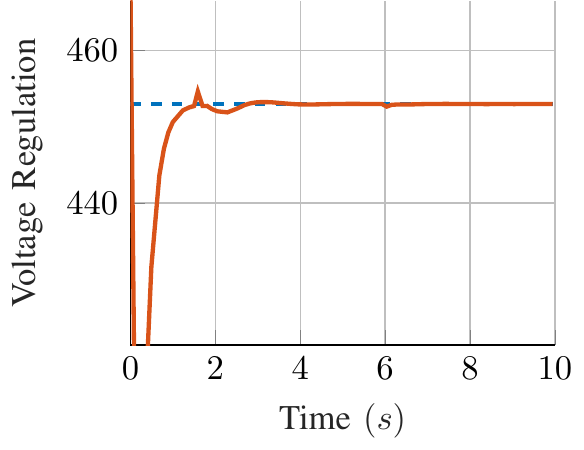}
		\caption{}
		\label{fig:ZIP_voltageregulation}
	\end{subfigure}
	\caption{PC voltages, weighted filter currents in \textit{per unit}, and weighted voltage sum under secondary control with ZIP loads. In (a), the black dashed lines represent the highest and lowest voltage values in $\mathcal{H}(\delta_{-})$, the set in which the unique equilibrium in $\mathcal{H}(\delta_{-})\cup\mathcal{I}$ lies.}
	\label{fig:ZIP}
\end{figure*}

\textbf{\textit{Initialization of the DCmG:}} First, the DCmG is initialized with all power lines and communication channels disconnected, i.e., there is no power transfer between the DGUs and the consensus-based controller is not activated. As such, the primary controllers of DGUs first regulate voltages at their PCs to corresponding reference voltages $V_{ref,i}$, as seen in Figure~\ref{fig:ZIP_voltages}. At this stage, DGUs $1$-$5$ supply ZIP loads, whereas DGU $6$ has a ZIE load with exponent $r_6 =0.65$. 

\textbf{\textit{Connection of DGUs:}} At $t=1.5s$, the switches on the power lines $l_1$, $l_2$, $l_4$, $l_5$, and $l_6$ are closed, connecting the DGUs $1$-$5$ to form a DCmG. Simultaneously, the consensus-based controller is activated with zero initial conditions for these DGUs. We would like to emphasize that, in this phase of the simulations, DGU $6$ is still disconnected from the rest of the DCmG. Theorems~\ref{thm:existence} and \ref{thm:stability} can be applied for this DCmG to conclude that there exists a unique equilibrium point for the PC voltages in $\mathcal{H}(\delta_{-})=\{V\in \mathbb{R}^5 |  (1-\delta_{-})V^*\leq V \leq(1+\delta_{-})V^*\}$, where $\delta_{-}=3.94\times 10^{-4}$ and $V^* = \left[47.15, 47.17, 47.18, 47.21, 47.26\right]^\top$. Moreover, this equilibrium point is the unique equilibrium point in $\mathcal{H}(\delta_{-})\cup \mathcal{I} = \{V\in \mathbb{R}^5 |  V\geq (1-\delta_{+})V^*\}$ with $\delta_{+}=0.9996$, and it is stable. Figure~\ref{fig:ZIP_voltages} shows that the PC voltages indeed converge to this unique equilibrium point in $\mathcal{H}(\delta_{-})$. Moreover, Figures~\ref{fig:ZIP_currents} and \ref{fig:ZIP_voltageregulation} respectively present that current sharing is achieved and voltages $V_i$ are successfully regulated to the references $V_{ref,i}-\omega_i$. 

\textbf{\textit{Change of ZIP loads:}} At $t=6s$, the ZIP loads in DGUs $1$ and $4$ are modified to increase their constant-impedance and constant-power loads, while still satisfying the conditions for Theorems~\ref{thm:existence} and \ref{thm:stability}. Consequently, it is guaranteed that a unique and stable equilibrium in $\mathcal{H}(\delta_{-})$ exists with $\delta_{-}=5.32\times 10^{-4}$ and $V^* = \left[47.13, 47.16, 47.17, 47.20, 47.28\right]^\top$. Furthermore, this is the unique equilibrium point in $\mathcal{H}(\delta_{-})\cup \mathcal{I}$ with $\delta_{+}=0.9995$. It can be seen in Figure~\ref{fig:ZIP} that voltages converge to this new equilibrium point with modified ZIP loads, as well as that current sharing and voltage regulation are achieved.
\begin{figure*}[t]
	\definecolor{mycolor1}{rgb}{0.00000,0.44700,0.74100}%
	\definecolor{mycolor2}{rgb}{0.85000,0.32500,0.09800}%
	\definecolor{mycolor3}{rgb}{0.92900,0.69400,0.12500}%
	\definecolor{mycolor4}{rgb}{0.49400,0.18400,0.55600}%
	\definecolor{mycolor5}{rgb}{0.46600,0.67400,0.18800}%
	\definecolor{mycolor6}{rgb}{0.30100,0.74500,0.93300}%
	\centering
	\setlength\fheight{3.5cm} 
	\setlength\fwidth{0.25\textwidth}
	\begin{subfigure}[t]{0.32\textwidth}
		\centering
		\includegraphics{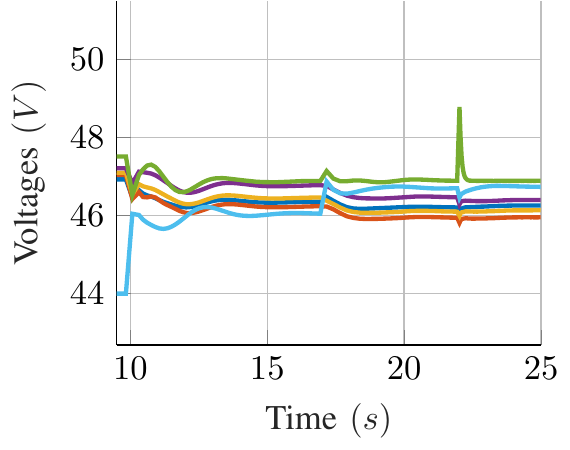}
		\caption{}
		\label{fig:ZIE_voltages}
	\end{subfigure}
	\begin{subfigure}[t]{0.32\textwidth}
		\centering
		\includegraphics{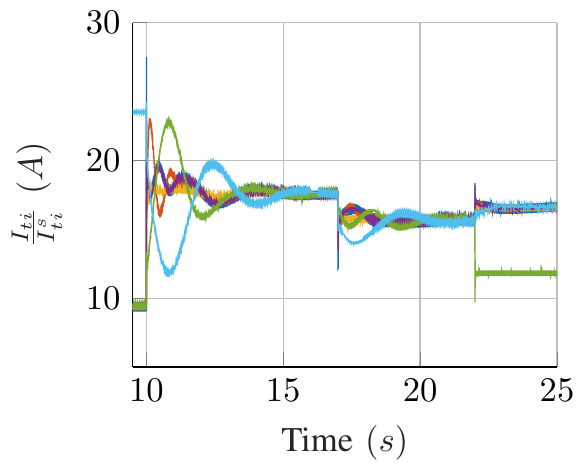}
		\caption{}
		\label{fig:ZIE_currents}
	\end{subfigure}
	\begin{subfigure}[t]{0.32\textwidth}
		\centering
		\includegraphics{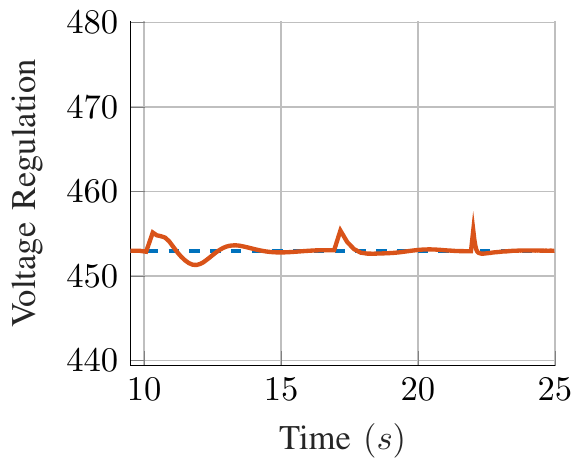}
		\caption{}
		\label{fig:ZIE_voltageregulation}
	\end{subfigure}
	\caption{PC voltages, weighted filter currents in \textit{per unit}, and weighted voltage sum under secondary control with ZIE loads.}
	\label{fig:ZIE}
\end{figure*}

\textbf{\textit{Plug-in of DGU $6$:}} At $t=10s$, the physical lines $l_3$ and $l_7$ are attached to connect the DGU $6$ to the rest of the DCmG. At the same time, the secondary controller of DGU $6$ is activated and those of DGUs $3$ and $5$ are updated to account for the communication from DGU $6$. Simultaneously, the constant-power loads of DGUs $2$, $3$, and $5$ are changed to exponential loads with exponents $r_2 = 0.6$, $r_3 = 0.55$, and $r_5 = 0.4$. Due to the existence of ZIE loads in the DCmG, Theorem~\ref{thm:existence} cannot be applied for this new DCmG; however, Theorem~\ref{thm:stability} can be applied to show that, if an equilibrium satisfying \eqref{eq:maximumpowerconsumption} exists, it is stable. Indeed, we see in Figures~\ref{fig:ZIE_voltages} and \ref{fig:ZIE_currents} that PC voltages converge towards an equilibrium point in positive orthant of $\mathbb{R}^N$, which results in current sharing. 

\textbf{\textit{Change of ZIE loads:}} At $t=17s$, exponents of the ZIE loads attached to DGUs $3$ and $6$ are changed to values greater than one, i.e., $r_3 = 1.45$ and $r_6 = 1.35$. This change of loads, in turn, change the dynamics of the DCmG, thus leading to a change of operation point. As can be seen in Figure~\ref{fig:ZIE}, the primary and secondary control objectives are satisfied under this load change.

\textbf{\textit{Unplugging of DGU $5$:}} At $t=22s$, in order to show that the proposed controller works under unplugging of DGUs from the DCmG, the DGU $5$ is isolated by opening the switches of lines $l_6$ and $l_7$. In doing so, its consensus-based controller is disabled, and those of its former neighbors, DGUs $4$ and $6$, are modified. Figure~\ref{fig:ZIE} shows that DGUs $1$, $2$, $3$, $4$, and $6$ achieve current sharing and voltage balancing, whereas DGU $5$ supplies its own load after unplugging from the DCmG. 


\subsection{Nonexistence of equilibria}
\label{subsec:SimulationZIPFailure}

Take into consideration the DCmG in Figure~\ref{fig:mG_simulation}, where all DGUs have ZIP loads. In this scenario, we modify only the line resistances in the DCmG to show that a solution to \eqref{eq:PFmatrix} may fail to exist if necessary conditions are not satisfied. In particular, we increase the line resistances such that the condition in \eqref{eq:Criticalpower} is no longer satisfied. Consequently, Theorem~\ref{thm:existence} can not be applied, meaning that the current sharing may not be achieved. 

\begin{figure}[t!]
	\centering
	\definecolor{mycolor1}{rgb}{0.00000,0.44700,0.74100}%
	\definecolor{mycolor2}{rgb}{0.85000,0.32500,0.09800}%
	\definecolor{mycolor3}{rgb}{0.92900,0.69400,0.12500}%
	\definecolor{mycolor4}{rgb}{0.49400,0.18400,0.55600}%
	\definecolor{mycolor5}{rgb}{0.46600,0.67400,0.18800}%
	\definecolor{mycolor6}{rgb}{0.30100,0.74500,0.93300}%
	\begin{tikzpicture} 
    \begin{axis}[%
    hide axis,
    xmin=10,
    xmax=50,
    ymin=0,
    ymax=0.4,
    legend style={legend columns=6, legend cell align=left, align=left, draw=white!15!black}
    ]
    \addlegendimage{mycolor1,line width = 1.5}
    \addlegendentry{\small{ $V_1$}};
    \addlegendimage{mycolor2,line width = 1.5}
    \addlegendentry{\small{ $V_2$}};
    \addlegendimage{mycolor3,line width = 1.5}
     \addlegendentry{\small{ $V_3$}};
     \addlegendimage{mycolor4,line width = 1.5}
    \addlegendentry{\small{ $V_4$}};
    \addlegendimage{mycolor5,line width = 1.5}
    \addlegendentry{\small{ $V_5$}};
    \addlegendimage{mycolor6,line width = 1.5}
     \addlegendentry{\small{ $V_6$}};
    \end{axis}
\end{tikzpicture}
		\includegraphics{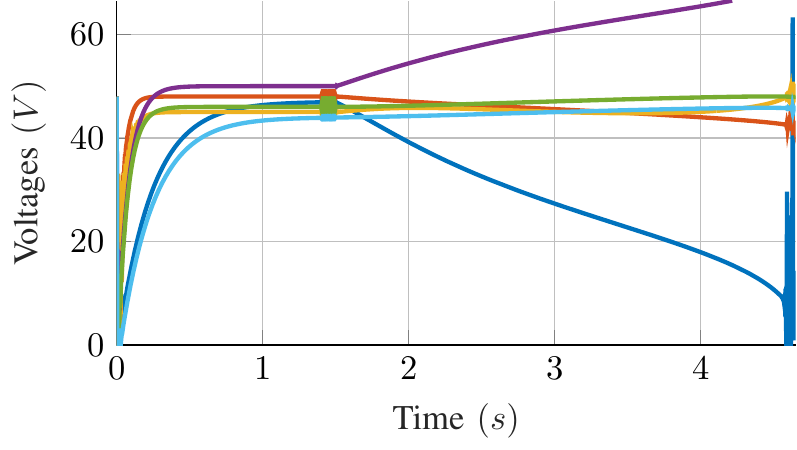}
	\caption{PC voltages with ZIP loads, when \eqref{eq:Criticalpower} is not satisfied.}
	\label{fig:Inex_voltages}
\end{figure}

To present this phenomenon through simulation, we initialize the DCmG with all the physical lines and communication channels disconnected, as in the first scenario above. Then, all electrical lines and communication channels are attached to connect all $6$ DGUs together at $t=1.5s$, also activating the secondary controllers. As can be seen in Figure~\ref{fig:Inex_voltages}, an equilibrium point does not exist. This results in a \textit{voltage collapse} in a short period of time, i.e., one of the PC voltages fall down to $0V$, which indicates an unsafe operating point where many electrical devices would either shut down or get damaged.

\section{Conclusions}
\label{sec:conclusions}
In this paper, a novel secondary consensus-based control layer for current sharing and voltage balancing in DCmGs was presented. We considered a DCmG composed of realistic DGUs, RLC lines, and ZIE loads. A rigorous steady-state analysis was conducted, and appropriate conditions ensuring the attainment of both objectives were derived. In addition,we provided a voltage stability analysis showing that the controllers can be synthesized in a decentralized fashion. Future work will study the impact of communication network non-idealities (such as transmission delays, data quantization and packet drops) on the performance of closed-loop mGs. Further developments can also consider the inclusion of Boost and other DC-DC converters.
\bibliographystyle{IEEEtran}
\bibliography{IEEEabrv,articleref}
\appendix
\begin{lemma}
	\label{lem:rangeLE}
	The range space of $\LL_p$ is $H^1$.
\end{lemma}
\begin{proof}
	Given that both $\LL_e$ and $\LL_t$ are symmetric Laplacian matrices of a connected graph, $\textbf{1}_N^T\LL_p=0$. Hence, $\RR(\LL_p)\subseteq H^1$. Note that the matrix $\LL_p^T$ is the Laplacian of a connected graph \cite{Bullo}, and therefore, $dim(\RR(\LL_p^T))=N-1$. Since $H^1$ is an $N-1$ dimensional subspace of $\mathbb{R}^N$ and $dim(\RR(\LL_p))=dim(\RR(\LL_p^T))=N-1$, $\RR(\LL_p) =H^1$.
\end{proof}
\subsection{LaSalle Analaysis}
\label{appendixA}
On invoking LaSalle's invariance principle, one has that, if $\tilde{X}(0)\in {\MM}$, then the state $\tilde{X}(t)$ asymptotically converges to the largest invariant set in \begin{equation}
\label{eq:SETE}
\begin{split}
{E}=\left\{\right. \tilde{X} \in {\MM}:	\dot{\VV}{(\tilde{X})}=0\left.\right\}.\\
\end{split}
\end{equation}
Now by \eqref{eq:lyapderivative}, $\dot{\VV}{(\tilde{X})}=0$ if and only if $\tilde{X}\in\NN(\QQ)$. By direct computation, the set $E$ can equivalently be represented in terms of the state $\tilde{X}$  as
\begin{equation}
\label{eq:setE}
E=\left\{\tilde{X}\in{\MM}~|~\tilde{X}=\begin{bmatrix}
p\\
[\gamma]q \\
-[\beta]q\\
\textbf{0}_M\\
s
\end{bmatrix}, q,s \in \mathbb{R}^{N} \right \},
\end{equation}
where $p\in \mathbb{R}^N$ when $Y_L+Y_E(V)=\textbf{0}$, otherwise $p=\textbf{0}_N$. For evaluating the largest invariant set in $E$, we pick the general case, that is, $p\in \mathbb{R}^N$.	In order to conclude the proof, we need to show that the largest invariant set $M\subseteq E\subseteq{\MM}$ is uniquely the equilibrium point $\bar{X}$. To find the largest invariant set, we aim to deduce conditions on $\tilde{X} \in E$ such that $\dot{\tilde{X}} \in E$. Using \eqref{eq:setE} and \eqref{eq:globalstatespace} we obtain
\begin{equation*}
\begin{split}
\dot{\tilde{X}}=\dot{X}&=\mathcal{A}\bar{X}+\mathcal{A}\begin{bmatrix}
p\\
[\gamma]q \\
-[\beta]q\\
\textbf{0}_M\\
s
\end{bmatrix}+\mathcal{B}(V)\\
&=\begin{bmatrix}
-C_t^{-1}(Y_{L}+Y_E(V))p+C_t^{-1}[\gamma]q\\
[\alpha]p+[\delta][I^s_t]^{-1}\mathcal{L}_cs\\
-p-[I^s_t]^{-1}\mathcal{L}_cs\\
L^{-1}B^Tp\\
\mathcal{L}_c[I^s_{t}]^{-1}[\gamma]q
\end{bmatrix}.
\end{split}
\end{equation*}
Therefore, $\dot{\tilde{X}} \in E$, if and only if $L^{-1}B^Tp=\textbf{0}_M$ and the following equations hold:
\begin{subequations}
	\begin{align}
	[\alpha]p+[\delta][I^s_t]^{-1}\mathcal{L}_cs&=[\gamma]\tilde{q}, \label{eq:E_inv_a}\\
	-p-[I^s_t]^{-1}\mathcal{L}_cs&=	-[\beta]\tilde{q}, \label{eq:E_inv_b}
	\end{align}
\end{subequations}
where $\tilde{q} \in \mathbb{R}^N$.Left multiplying \eqref{eq:E_inv_b} with $[\alpha]$, and then adding it with \eqref{eq:E_inv_a} yields
\begin{equation}
[\alpha][\beta]q=[\gamma]q.
\end{equation}
This necessitates 
\begin{equation}
\label{eq:condition}
\alpha_i\beta_i=\gamma_i ,\forall i \in \DD.
\end{equation}
Also, as the feedback gains $k_{1,i},k_{2,i},$ and $k_{3,i}$ belong to the set $\ZZ_{[i]}$ in \eqref{eq:explicitgains}, then $\alpha_i<0,\beta_i<0$, and $0<\gamma_i<\alpha_i\beta_i$. Thus, we conclude that \eqref{eq:E_inv_a} and \eqref{eq:E_inv_b} can be simultaneously satisfied only if $\tilde{q}=\textbf{0}_{N}$.  As for $L^{-1}B^Tp=\textbf{0}_M$, one obtains that  $a \in \NN(B^T)$. Since the graph $\GG$ is connected, $\NN(B^T)=H_{\perp}^1$ \cite{Bullo}. Therefore, for $\dot{\tilde{X}}$ to remain in $E$, $\tilde{X}$ must stay in set $S \subset E$, where 
\begin{equation}
S=\left\{\tilde{X}\in{\MM}~|~\tilde{X}=\begin{bmatrix}
\kappa\textbf{1}_N\\
\textbf{0}_N \\
\textbf{0}_N\\
\textbf{0}_M\\
s
\end{bmatrix}, s \in \mathbb{R}^{N} \right \}.
\end{equation}
Furthermore, it must hold $M \subseteq S$. Then, in order to characterize $M$, we assume $\tilde{X} \in S$ and impose $\dot{\tilde{X}} \in S$. This translates into the following
\begin{equation*}
\begin{split}\dot{\tilde{X}}=\dot{X}&=\mathcal{A}\bar{X}+\mathcal{A}\begin{bmatrix}
\kappa\textbf{1}_N\\
\textbf{0}_N \\
\textbf{0}_N\\
\textbf{0}_M\\
s
\end{bmatrix}+\mathcal{B}(V)\\&=\begin{bmatrix}
-{\kappa}C_t^{-1}(Y_{L}+Y_E(V))\textbf{1}_N\\
{\kappa}[\alpha]\textbf{1}_N+[\delta][I^s_t]^{-1}\mathcal{L}_cs\\
-{\kappa}\textbf{1}_N-[I^s_t]^{-1}\mathcal{L}_cs\\
\textbf{0}_M\\
\textbf{0}_N
\end{bmatrix}.\end{split}\end{equation*}
Notice that, for $\dot{\tilde{X}} \in S$, it must hold that $[I^s_t]^{-1}\mathcal{L}_cs=-{\kappa}\textbf{1}_N$. Since ${\kappa}\textbf{1}_N\in \NN([I^s_t]^{-1}\mathcal{L}_c)$, it turns out that both $\kappa=0$ and $s=\mbf{0}_N.$ This implies that the largest invariant set $M\subseteq E$ is $M= \{\tilde{X}\in{\MM}~|~\tilde{X}=\textbf{0}_{4N+M}\}$. 

\end{document}